\documentclass{article}

\usepackage[utf8]{inputenc}

\usepackage{etoolbox}
\newtoggle{usingbibtex}
\togglefalse{usingbibtex}

\iftoggle{usingbibtex}{
\usepackage[square,numbers]{natbib}
\bibliographystyle{unsrtnat}
}{
\usepackage[
backend=biber, %style=alphabetic,%doi=false, %sorting=ynt,%doi=true,%eprint=true,%language=false,
style=numeric, sorting=none, isbn=false,eprint=false,url=false,
]{biblatex}
\AtEveryBibitem{%
  \clearlist{language}%
}
\addbibresource{library.bib}
}

\usepackage{times}
\usepackage{graphicx}
\usepackage{subfigure} 
\usepackage{bm}

\usepackage{algorithm}
\usepackage[noend]{algorithmic}
\usepackage{amsmath}
\usepackage{amssymb}
\usepackage{amsthm}
\usepackage{mathtools}
\usepackage{breqn}

\usepackage{siunitx}

\usepackage{multirow}
\usepackage{tabularx,booktabs}

\usepackage[section]{placeins}

\usepackage{xcolor}

\usepackage{xspace}
\usepackage[shortlabels]{enumitem}
\usepackage{fmtcount}

\usepackage{hyperref}
\let\proglang=\textsf

\newcommand{\code}[1]{{\texttt {#1}}}
\let\proglang=\text

\newcommand{\figref}[1]{Fig.~\ref{fig:#1}}
\newcommand{\tblref}[1]{Table~\ref{tab:#1}}
\newcommand{\sref}[1]{Sect.~\ref{#1}}
\newcommand{\eqnref}[1]{Eq.~\eqref{eq:#1}}
\newcommand{\eqnsref}[2]{Eqs.~\eqref{eq:#1}, \eqref{eq:#2}}
\newcommand{\formularef}[1]{Formula~\ref{eq:#1}}
\newcommand{\formulasref}[2]{Formulas~\ref{eq:#1},~\ref{eq:#2}}
\newcommand{\algref}[1]{Algorithm~\ref{alg:#1}}
\newcommand{\lemmaref}[1]{Lemma~\ref{lemma:#1}}
\newcommand{\theoremref}[1]{Theorem~\ref{theorem:#1}}
\newcommand{\stepref}[1]{~\ref{#1}}
\newcommand{\algabbrevref}[1]{Alg.\@~\ref{alg:#1}}

\newcommand*{\eg}{e.g.\@\xspace}
\newcommand*{\Ie}{I.e.\@\xspace}
\newcommand*{\Eg}{e.g.\@\xspace}

\newcommand{\doi}[1]{\href{http://dx.doi.org/#1}{\normalfont\texttt{\@doi{#1}}}}

\newcommand{\Prob}{\ensuremath{\mathsf{P}}}
\newcommand{\CDF}{\ensuremath{\mathsf{CDF}}}
\newcommand{\SF}{\ensuremath{\mathsf{SF}}}
\newcommand{\PDF}{\ensuremath{\mathsf{PDF}}}
\newcommand{\PPF}{\ensuremath{\mathsf{PPF}}}
\newcommand{\ISF}{\ensuremath{\mathsf{ISF}}}
\def\SciPy{SciPy }

\newcommand{\pow}{\mathop{\mathrm{pow}}}

\newcommand{\loggamma}{\mathop{\mathbin{{\log}{\Gamma}}}}

\newcommand{\smirnov}{\mathop{\mathrm{smirnov}}}
\newcommand{\smirnovi}{\mathop{\mathrm{smirnovi}}}
\newcommand{\smirnovcode}{\code{smirnov}}

\newtheorem{pvmalgorithm}{Algorithm}
\newtheorem{theorem}{Theorem}[section]

\newtheorem{lemma}[theorem]{Lemma}

\DeclareMathOperator{\powfour}{powFour}
\DeclareMathOperator{\logonep}{log1p}

\DeclareMathOperator{\modf}{modf}
\DeclareMathOperator{\logBinomial}{logBinomial}
\DeclareMathOperator{\lbeta}{lbeta}
\DeclareMathOperator{\fl}{fl}

\newcommand{\keywords}[1]{\par\addvspace\baselineskip
\noindent\keywordname\enspace\ignorespaces#1}
\ifx\mycmd\keywordname
\def\keywordname{{\bf Keywords:}}
\else
\fi

    \title{Computing the Cumulative Distribution Function and Quantiles of the One-sided Kolmogorov-Smirnov Statistic}
    
    \author{Paul van Mulbregt\\
pvanmulbregt@alum.mit.edu
    }
    
\begin{document} 
   
    \maketitle
%    \vskip 0.3in

    %!TEX root = ./ms.tex
 \begin{abstract}
The cumulative distribution and quantile functions for the 
one-sided one sample Kolmogorov-Smirnov probability distributions 
are used for goodness-of-fit testing. 
While the Smirnov-Birnbaum-Tingey formula for the \CDF \/ appears straight forward, its numerical evaluation generates intermediate results spanning many hundreds of orders of magnitude and at times requires very precise accurate representations.  
Computing the quantile function for any specific probability may require evaluating both the \CDF\/ and its derivative, both computationally expensive. 
To work around avoid these issues, different algorithms can be used across different parts of the domain, and
approximations can be used to reduce the computational requirements.
We show here that straight forward implementation incurs accuracy loss for sample sizes of well under 1000.
Further the approximations in use inside the open source SciPy python software often result in increased computation, not just reduced accuracy, and at times suffer catastrophic loss of accuracy for any sample size.
Then we provide alternate algorithms which restore accuracy and efficiency across the whole domain.
\end{abstract}

    %!TEX root = ./ms.tex

\section{Introduction}

The Kolmogorov-Smirnov statistics $D_n$, $D_n^+$, $D_n^-$ are statistics that can be used as a measure of the goodness-of-fit between a sample of size $n$ and a target probability distribution.  
Computation of the exact probability distribution for these statistics is not a little complicated, but Kolmogorov and Smirnov showed that they had a certain limiting behaviour as $n \rightarrow \infty$.
To be used as part of a statistical test, either the value of the Survival Function (\SF) (or its complement the \CDF) 
needs to computable for a given value of $D_n$(/$D_n^+$), or values need to be 
known corresponding to the desired critical probabilities (\Eg $p=0.1, 0.01, \ldots$).  
The quantile functions associated with these distributions can be used to generate a table fo critical values, but they can also be used to generate random variates for the distribution, and also found applications to rescaling of the axes in some kinds of plots.
For the one-sided $D_n^+$, a formula is known which can be used to compute the \SF\@. 
The quantile function does not have a closed-form solution, hence needs to be calculated either by interpolating some known values, or by a numerical root-finding approach.

The \proglang{Python} package \SciPy  (v0.19.1) \cite{SciPy} provides the \code{scipy.stats.ksone} class for the distribution of the one-sided $D_n^+$ which in turn make calls to the ``C" library \code{scipy.special}, to calculate both the \SF\/ and \ISF\@.
An analysis of the implementation determined that a straight translation of the \SF\/ formula into ``C" code is likely to incur severe accuracy loss.

The outcome is that underflow and/or denormalization occurring in intermediate results, so that the computed values could differ for the correct ones by several orders of magnitude.

The Inverse Survival Function (\ISF) is computed by the Newton-Raphson root finding algorithm. 
As the computation of the \SF\/ (and by implication its derivative the \PDF), the \ISF\/ uses several approximations to shorten the calculation.  
However the range of the approximations is not the full domain needed.  The result is loss of accuracy and/or increased computation, and in some cases complete failure.  These issues appear not just for large $n$, but even for $n=1$ or $2$.

In this paper we analyze the approximations in the algorithms and the implementation of the algorithms.
The causes of accuracy loss and other behaviours are identified.  Several causes of root-finding failure are identified.
We then provide alternate algorithms which have much lower relative error as well as lower (and bounded) computation.
For the quantile functions, the number of Newton-Raphson iterations is reduced by a factor of 2-3, and the maximum by about 50.

This paper is organized as follows.  
\sref{sec:ks_review} provides a quick review of Kolmogorov-Statistics with special emphasis on the formulae needed for computation.
\sref{sec:smirnov} analyses the formulae for computing the $\CDF$/$\SF$ of the one-sided $D_n^+$, and the \SciPy  implementation,
\sref{sec:smirnov_prop} details an alternate recipe for computing the $\CDF$/$\SF$, along with an error analysis of the algorithm.
\sref{sec:smirnovi} analyses the computation of the \ISF, determining reasons for convergence failures.
\sref{sec:smirnovi_prop} provides an alternate recipe, a formula for a narrow interval which encloses the root, and a close initial estimate
together with analysis of these formulae.
\sref{sec:smirnov_results} provides numeric results showing the change in performance resulting from use of these algorithms, 
along with interpretation of results. 

The formulae for computing the $\SF$/$\CDF$/$\PDF$ have been available for quite some time.  
The novelty in this work is the analysis of the \SciPy  implementations, and the details of the recipes, especially for the quantile functions.  

 The ``C" code computing the $\CDF$ \& $\SF$ for this distribution was written quite some time ago, 
 when computers had considerably slower clock speeds and sample sizes were considerably smaller than they are today.
 To a user of the software, the answers may have seemed plausible for most real-world size inputs.
 Not many samples have only 1 or 2 elements, where some of the issues are most prevalent.
 A user with a sample of size 1000 may not have realized that some small probabilities were much too small.

    %!TEX root = ./ms.tex

\section{Review of Kolmogorov-Smirnov Statistics}
\label{sec:ks_review}

In 1933 Kolmogorov \cite{kolmogoroff1933, kolmogoroff1941} introduced the the Empirical Cumulative Distribution Function (ECDF) for a (real-valued) sample $\{Y_1, Y_2, \dots, Y_n\}$, each $Y_i$ having the same continuous distribution function $F(Y)$.
He then enquired how close this ECDF would be to $F(Y)$.   Formally he defined
\begin{align}
F_n(y) & = \frac{1}{n} \#\left\{i : Y_i <= y\right\} \\
D_n & = \sup_y \left|F_{n}(y)-F(y)\right|
\end{align}

\figref{KSconstruction} illustrates the ECDF, and the construction of $D_n^+$ and $D_n^-$ for one sample.

After wondering whether $\Prob \{D_n \leq \epsilon\}$ tends to 1 as $n\rightarrow\infty$ for all $\epsilon$, he then answered affirmatively with the asymptotic result  \cite{kolmogoroff1933, kolmogoroff1941}
\begin{align}
\lim_{n\rightarrow\infty} \Prob \{D_n \leq x\, n^{-1/2}\} & = L(x) = 1 - 2 \sum_{k=1}^{\infty} (-1)^{k-1} e^{-2k^2x^2}
\end{align}
Kolmogorov's proof used methods of classical physics.  Feller \cite{feller1948, feller1950} provided a more accessible proof in English.

Smirnov \cite{zbMATH03107305, smirnov1948} instead used the one-sided values
$D_n^+ = \sup_y \left(F_{n}(y)-F(y)\right)$ and $D_n^- = \sup_y \left(F(y)-F_{n}(y)\right)$
 and showed that they also had a limiting form
\begin{align}
\label{eq:SmirnovAsymptote}
\lim_{n\rightarrow\infty} \Prob \{D_n^+ \leq x\, n^{-1/2}\} & = \lim_{n\rightarrow\infty} \Prob \{D_n^{-} \leq x\, n^{-1/2}\} = 1 - e^{-2x^2} \\
\shortintertext{Magg \& Dicaire\cite{Maag1971} gave a tightened asymptotic.  For a fixed $x$}
\label{eq:MaagAsymptote}
\Prob \{D_n^+ \leq x\} & \underset{{n\rightarrow \infty}}{\asymp}  1 - \exp\left({\frac{-(6nx+1)^2}{18n}}\right)
\end{align}

\begin{figure}[!htb]
\centering
  \includegraphics[scale=0.25]{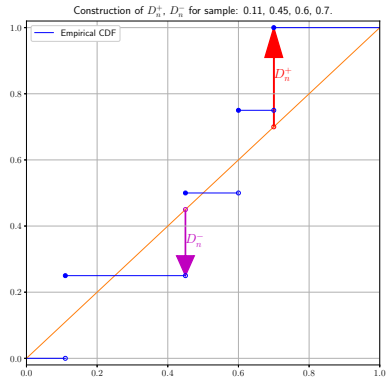}
 \caption{Construction of Kolmogorov-Smirnov statistics for $n=4$.}
\label{fig:KSconstruction}
\end{figure}

For the purpose of showing that the ECDF approaches $F(Y)$, these limit formulae are sufficient.
Later authors turned this around and used the $D_n$ statistic as a measure of ``goodness-of-fit" between the sample and $F(Y)$, for {\it any}\/ distribution function $F(Y)$.
It is clear that a large value for any of $D_n$, $D_n^+$ and $D_n^-$  may be indicative of a mismatch.
But a too-small value may also be cause for concern, as the fit may ``too good".
In order to use $D_n$ for this goodness-of-fit purpose, knowledge of the distribution of $D_n$ itself is needed, not just its limit as $n \rightarrow \infty$.

Determining the exact distribution of the two-sided $D_n$ is non-trivial.
Birnbaum \cite{birnbaum1952} showed how to use Kolmogorov's recursion formulas to generate exact expressions for $\Prob(D_n \leq x)$, $n\leq6$.
Durbin \cite{durbin1968} provided a recursive formula to compute $\Prob(D_n \leq x)$ (implemented by Marsaglia, Tsang and Wang \cite{JSSv008i18}, made more efficient by Carvalho \cite{JSSv065c03}), which involved calculating a particular entry in a potentially large matrix raised to a high power.
Pomeranz \cite{Pomeranz:1974:AEC:361604.361628} provided another formulation which involved calculating a specific entry in a large-dimensional matrix.
Drew Glen \& Leemis \cite{Drew:2000:CCD:350528.350529} generated the collection of polynomial splines for $n<=30$. 
Brown and Harvey \cite{JSSv019i06, JSSv026i02, JSSv026i03} implemented several algorithms in both rational arithmetic and arbitrary precision arithmetic.
Simard and L'Ecuyer \cite{JSSv039i11} analyzed all the known algorithms for numerical stability and sped.

For the one-sided statistics the situation is much cleaner.
An exact formula was discovered early-on \cite{zbMATH03107305, birnbaum1951, Miller1956}
 \begin{align}
 \label{eq:sm_1a}
 \Prob(D_n^+ \leq x) & = 1 - S_n(x)\\
 \shortintertext{where}
 \label{eq:sm_1}
 S_n(x) & =  x\sum_{j=0}^{\lfloor n(1-x)\rfloor} \binom{n}{j}  \left(x+\frac{j}{n}\right)^{j-1} \left(1-x-\frac{j}{n}\right)^{n-j}  \\
 & = x\sum_{j=0}^{\lfloor n(1-x)\rfloor} A_j(n, x)  \\
 A_j(n, x) & =  \binom{n}{j}  \left(x+\frac{j}{n}\right)^{j-1} \left(1-x-\frac{j}{n}\right)^{n-j}
 \label{eq:Ajn}
 \end{align}
$S_n(x) $ is a sum of relatively simple $n$-th degree polynomials, forming a spline with knots at $0, \frac{1}{n}, \frac{2}{n}, \dots, 1$.
The first two splines are:
\begin{align}
\label{eq:S2}
S_1(x) & =  \left(1-x\right)\\
S_2(x) & = \left\{
    \begin{array}{lr}
    \left(1-x\right)^2 +2x(\frac{1}{2}-x), & \text{for } 0\leq x\leq \frac{1}{2}\\
    \left(1-x\right)^2,                             & \text{for } \frac{1}{2}\leq x\leq 1\\
    \end{array}
\right.
\end{align}

\begin{figure}[!htb]
\centering
\includegraphics[scale=0.5]{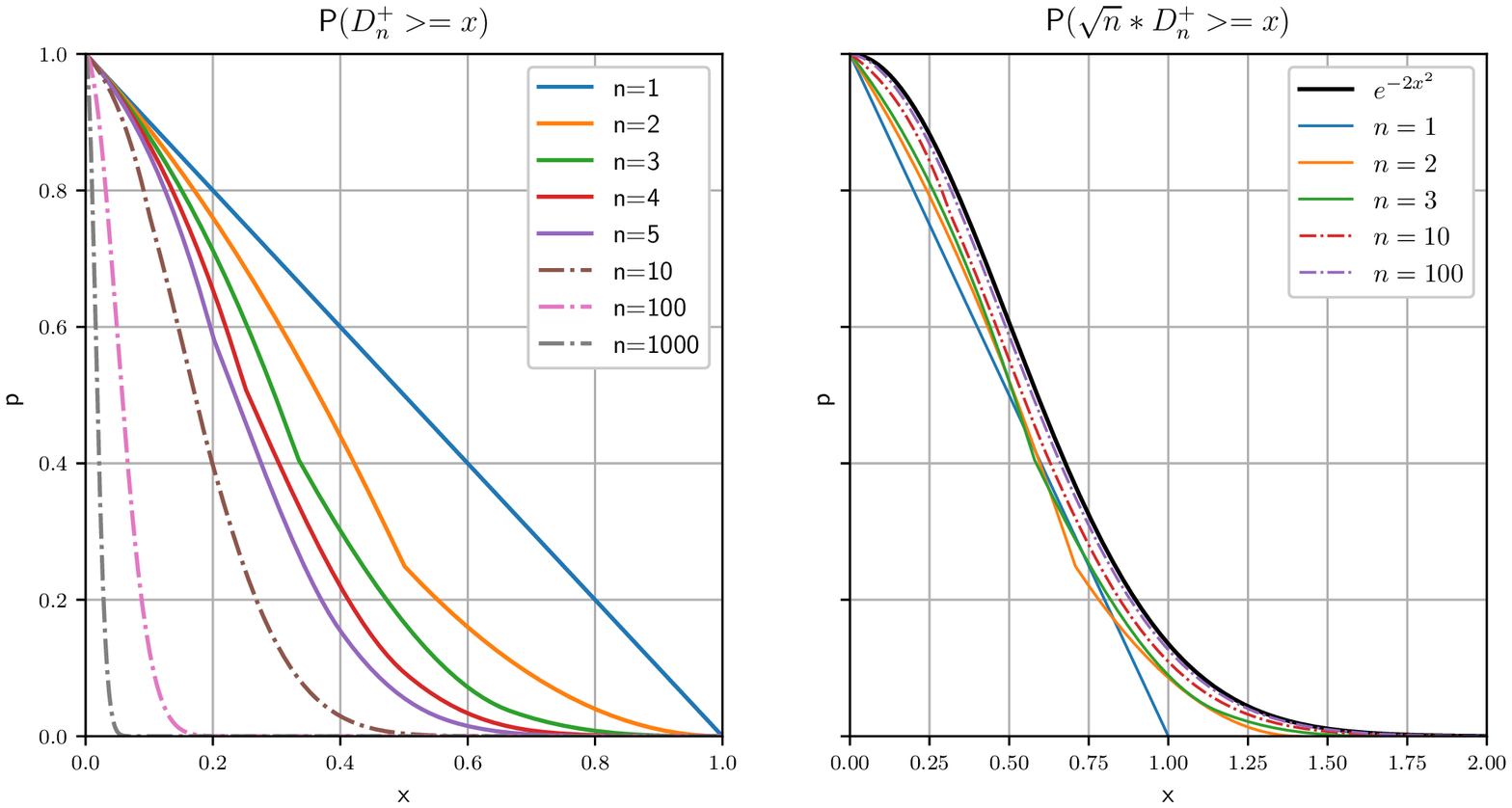}
\caption{$\Prob(D_n^+ >= x)$ and $\Prob(\sqrt{n}D_n^+ >= x)$}
\label{fig:smirnov_and_sqrtn}
\end{figure}

\figref{smirnov_and_sqrtn} shows the
survival probability distributions of $D_n^+$ for a few values of $n$,
and the probability distributions of $\sqrt{n}D_n^+$, together with the limit $e^{-2x^2}$.

There is no closed form solution for inverting $S_n(x)$.
The distributions of $D_n^+$ and $D_n^-$ are the same but they are not independent.
The distributions of $D_n$ and $D_n^+$ are also related:  $D_n \geq x \iff $ one or both of $D_n^+$, $D_n^{-}\geq x$.
In particular $\Prob(D_n \geq x) = 2 \Prob(D_n^+ \geq x)$ for $x\geq 0.5$.

    %!TEX root = ./ms.tex

\section[Computation of the Survival Function, smirnov()]{Computation of the Survival Function, $\smirnov()$}

\label{sec:smirnov}

The \textrm{scipy.special} subpackage of the \textrm{Python} \textrm{SciPy} package provides two functions for computations of the $D_n^+$ distribution.
\code{smirnov(n, x)} computes the Survival Function $S_n(x)$ for $D_n^+$
and \code{smirnovi(n, p)} computes the Inverse Survival Function.
The source code for the computations is written in ``C", and are performed using the IEEE 754 64 bit double type (53 bits in the significand, and 11 bits in the exponent.)
We'll let $\mathbb{F}$ be a set of radix-2 floating point numbers with precision $p$, and $\epsilon=2^{-p}$.

%%%%%%%%%%%%%%%%%%%%%%%%%%%%%%%%%%%%%%%%%%%%%%%%%%%
%%%%%%%%%%%%%%%%%%%%%%%%%%%%%%%%%%%%%%%%%%%%%%%%%%%
%%%%%%%%%%%%%%%%%%%%%%%%%%%%%%%%%%%%%%%%%%%%%%%%%%%

%%%%%%%%%%%%%%%%%%%%%%%%%%%%%%%%%%%%%%%%%%%%%%%%%%%
%%%%%%%%%%%%%%%%%%%%%%%%%%%%%%%%%%%%%%%%%%%%%%%%%%%
%%%%%%%%%%%%%%%%%%%%%%%%%%%%%%%%%%%%%%%%%%%%%%%%%%%

The Survival Function $S_n(x)$ can be computed using the Smirnov/Birnbaum-Tingey  \formularef{sm_1}. Let
$\pow(z, m)  = z^m$ and
\begin{align}
C_{n, j} & = \begin{cases}
    1, & \text{for } j=0\\
    C_{n, j-1} * (n-j+1)/j, & \text{for } j>0
    \end{cases}
\shortintertext{Then}
\label{eq:sm_2alpha}
S_n(x) & =  x *  \sum_{j=0}^{\lfloor n(1-x)\rfloor}  C_{n,j}* \pow(x+\frac{j}{n}, {j-1}) *
    \pow(1-(x+\frac{j}{n}), {n-j})
\end{align}
This gives an algorithm to compute $S_n(x)$.
Computing the survival probability $p$ for any particular $n, x$ pair involves summing at most $n$ triple products, each of the products being a binomial coefficient
and two powers of real numbers between 0 and 1.
If $x = \frac{n-j}{n}$ is one of the knots, then the top term in the summation need not be included as $A_j(n, x)$ has a zero of order $n-j$ at $x=\frac{n-j}{n}$.
Hence the upper limit in \eqnref{sm_2alpha} can be replaced by $\lceil n(1-x)\rceil - 1 = n - 1 - \lfloor nx\rfloor$.

%%%%%%%%%%%%%%%%%%%%%%%%%%%%%%%%%%%%%%%%%%%%%%%%%%%
%%%%%%%%%%%%%%%%%%%%%%%%%%%%%%%%%%%%%%%%%%%%%%%%%%%
%%%%%%%%%%%%%%%%%%%%%%%%%%%%%%%%%%%%%%%%%%%%%%%%%%%

%%%%%%%%%%%%%%%%%%%%%%%%%%%%%%%%%%%%%%%%%%%%%%%%%%%
%%%%%%%%%%%%%%%%%%%%%%%%%%%%%%%%%%%%%%%%%%%%%%%%%%%
%%%%%%%%%%%%%%%%%%%%%%%%%%%%%%%%%%%%%%%%%%%%%%%%%%%

\subsection{Controlling Accuracy Loss}

The calculation of \eqnref{sm_2alpha} is susceptible to accuracy loss when not using infinite precision arithmetic.
Just computing a single term $(x+\frac{j}{n})^{j-1} (1-x-\frac{j}{n})^{n-j} $ offers several opportunities to lose accuracy:
in the addition/subtraction, the exponentiation, and the multiplication, before considering the summation.
Brown and Harvey \cite{JSSv026i03} go into quite some detail on the precision required for internal computations in order to achieve a desired accuracy in the final result.
The magnitude of the terms in \eqnref{sm_2alpha} can span quite a range, but it is also that case that many terms in the sum may have similar magnitude.
 \figref{Ajn}(a) and \figref{Ajn}(b) shows the general phenomena.
For $x\geq\frac{1}{\sqrt{n}}$
there is a rapid rise for small $j$
 followed by a long slightly descending plateau, and a dramatic drop as $j$ approaches the end of the summation range.
This implies that many terms have a similar impact on the final sum, and all need to be computed and summed.
Computing powers amplifies rounding errors, proportional to the exponent.
Here the exponent can be quite large, a few hundred or thousands.
As most $x$ of interest are in the range $[0, \frac{2}{\sqrt{n}}]$ ($\approx 85\%$ are less than $\frac{1}{\sqrt{n}}$),
just computing $1\pm x$ can easily lose $\log_2(n)/2$ bits of $x$, even before the exponentiation.

\begin{figure}[!htb]
\begin{center}
\includegraphics[scale=0.5]{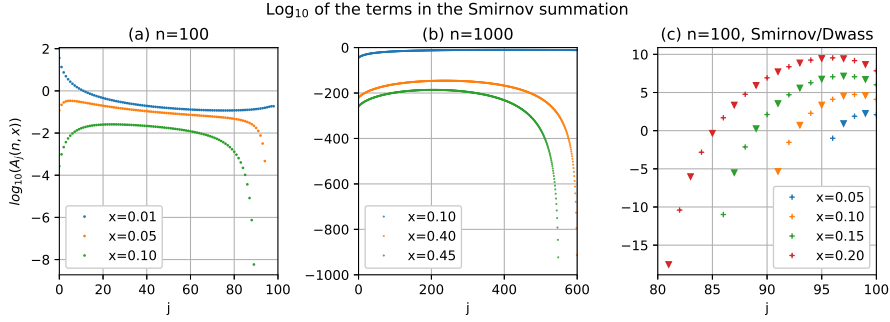}
\caption{Log of $A_j(n, x)$, the individual terms in the Smirnov summations for (a) $n=100$; and (b) $n=1000$. (c) $n=10$ shows terms from the alternate summation \eqnref{Sn_upper};  terms to be added are shown with a $+$, terms subtracted with $\blacktriangledown$. }
\label{fig:Ajn}
\end{center}
\end{figure}

To control roundoff errors in the individual terms, the calculation \eqnref{sm_2alpha} can be written
\begin{align}
\label{eq:sm_2p}
S_n(x) & =   x \sum_{j=k}^{n-1}  C_{n,j-k}* \left(\frac{\alpha+j}{n}\right)^{j-k-1} *
    \left(\frac{n-j-\alpha}{n}\right)^{n-j+k}
\end{align}
where $k$ and $\alpha$ are the integer and fractional parts of $nx$, $x=\frac{k+\alpha}{n}$, after noting that $\lceil n(1-x)\rceil = n - \lfloor{nx}\rfloor$.
Rewriting in this way allows any cancellation due to the integer subtractions to occur before the fractional $\alpha$ is considered.
For moderate sized $n$, the first obvious powers that may present some
floating point difficulty are the terms $(\frac{\alpha+j}{n})^{j-k-1}$ if $\alpha+j$ is
close to $n$ (so $j = n-1$), and $(1-\frac{\alpha+j}{n})^{n-j+k}$ if $\alpha+j$ is close to 0 (hence $j=0$).
The former is
\begin{align}
    \left(\frac{\alpha+j}{n}\right)^{j-k-1}  &= \left(1-\frac{1-\alpha}{n}\right)^{n-k-2}\\
     \shortintertext{and the latter is (note that $j=0 \implies k=0$ and $nx=\alpha$)}
    \left(1-\frac{\alpha+j}{n}\right)^{n-j+k}  &= \left(1-\frac{\alpha}{n}\right)^{n} = (1-x)^n
\end{align}
and these operations are well understood.

Using C's standard library function $\pow((x+j)/n, m)$ to compute $(x+\frac{j}{n})^m$ loses accuracy.
The relative error in using \code{pow} to calculate $z^m$ can be approximately determined.
\begin{align}
{\code{pow(fl(z), m)}} & = ({\code{fl(z)}})^m \cdot (1+\delta_2)\\
  & = (z(1+\delta_1))^m \cdot (1+\delta_2) \\
\shortintertext{where  $\delta_2$ is the relative error arising from the {\code{pow}} function, and
  $\delta_1= ({\code{fl}}(z)-z)/z$ is the relative error in ${\code{fl}}(z)$}
%$\delta_1 = ({\code{fl}}(z)-z)/z$}
 & \approx z^m (1 + m \delta_1) (1 + \delta_2)
\end{align}
The relative error in using $\log$/$\logonep$ and \code{exp} to calculate $z^m$ can also be estimated.
\begin{align}
     \code{exp(log(z)*m)} & = e^{\code{fl(log(z)*m)}} \cdot (1+\delta_4)\\
     & = e^{m \log(\code{fl(z)})\cdot(1+\delta_3) } \cdot (1+\delta_4)\\
     & = e^{m \log(z(1+\delta_1))\cdot(1+\delta_3) } \cdot (1+\delta_4)\\
    \shortintertext{where  $\delta_4$ and $\delta_3$ are the relative errors incurred using the $\exp$ and $\log$ functions respectively
    %, and $\delta_1$ is the relative error in $\text{fl}(z)$; $\delta_1 = (\code{fl(z)}-z)/z$
    }
     & \approx z^m (1 + m\log(z)\delta_3 + m \delta_1 +m \delta_3\delta_1)(1+\delta_4)\\
     & \approx z^m (1 + m \delta_1) (1 + m\log(z)\delta_3 + \delta_4)
\end{align}
Typically $\delta_2, \delta_3, \delta_4$ are on the order $\pm$0.5-2 ULP and can be treated as random with a mean of 0 (though $\delta_2$, the error in \code{pow}, may have less of a guarantee.)
$\delta_1$ however is fixed and known, hence contributes a bias to the computation which needs to be addressed.
That still leaves a relative error for each term in the summation of $\pm$1-4 ULP for the power computation, or much bigger if using logs, and there are $\approx n$ terms to add.
There is need for a function to compute $\left(\frac{x + j}{n}\right)^m$ more accurately.
One can use a compensated power algorithm as in \cite{graillat:ieee2009}.

The terms for $j$ in the middle range are in some ways more problematic.
If $j=\lfloor{n/2}\rfloor$, then the $j$-th term is approximately $\binom{n}{n-j-k} \frac{1}{2^{j-k-1}}\frac{1}{2^{n-j+k}}$.
The binomial term is very large, and the two powers are very small, all of which require accurate computation.

%%%%%%%%%%%%%%%%%%%%%%%%%%%%%%%%%%%%%%%%%%%%%%%%%%%
%%%%%%%%%%%%%%%%%%%%%%%%%%%%%%%%%%%%%%%%%%%%%%%%%%%
%%%%%%%%%%%%%%%%%%%%%%%%%%%%%%%%%%%%%%%%%%%%%%%%%%%

%%%%%%%%%%%%%%%%%%%%%%%%%%%%%%%%%%%%%%%%%%%%%%%%%%
%%%%%%%%%%%%%%%%%%%%%%%%%%%%%%%%%%%%%%%%%%%%%%%%%%%
%%%%%%%%%%%%%%%%%%%%%%%%%%%%%%%%%%%%%%%%%%%%%%%%%%%

\subsection{Overflow and Underflow}

Each term in the \formularef{sm_2alpha} is a triple product, of a positive integer and two powers of real numbers between 0 and 1.
The binomial factor can grow quite large, and will eventually overflow 64-bit integers near $n=66$, and overflow 64-bit floats around $n=1028$.
On the other hand, the powers may underflow, especially $(1-x-\frac{j}{n})^{n-j} $.
The presence of an underflow in \eqnref{sm_2alpha} does {\em{not}}\/ mean that the whole product would underflow.  In fact, for large $n$, it is likely the case that the most important terms in the summation all experience underflow!
Ignoring them leads to answers which are incorrect by orders of magnitude.

$A_j(n,x)$ can be written in several ways, one of which may aid in its computation.
\begin{align}
A_j(n, x)
& = \binom{n}{j} \left( \frac{x+\frac{j}{n}} {1-x-\frac{j}{n}} \right)^{j-1}   \left(1-x-\frac{j}{n}\right)^{n-1} \\
& = \binom{n}{j} \left(x+\frac{j}{n}\right)^{n-1}  \left( \frac{1-x-\frac{j}{n}}{x+\frac{j}{n}} \right)^{n-j}\\
& = \frac{(n-j+1)(nx+j)}{j(n-j-nx)} \left(1+\frac{1}{j-1 + nx}\right)^{j-2} \times \notag \\
\label{eq:Ajniterative}
 & \quad\quad\quad \left(1- \frac{1}{n-j+1-nx}\right)^{n-j-1} \,  A_{j-1}(n, x)\\
& = \frac{(n-j+1)(nx+j)}{j(n-j-nx)} \; A_{j-1}(n, x+\frac{1}{n}) \\
\label{eq:Ajnlog}
& = \exp\bigg(\log\binom{n}{j} +  ({j-1}) \log(x+\frac{j}{n}) + ({n-j})\log(1-x-\frac{j}{n})\bigg)
\end{align}
The applicability of the first two is somewhat limited, due to the exponent $n-1$, which tends to result in underflow.
 \eqnref{Ajniterative} forms the basis for an iterative approach to evaluation \cite{zbMATH03107305,dwass1959,JSSv019i06}.
 \eqnref{Ajnlog} can be used in some situations, but subtractive cancellation can occur inside the argument to $\exp()$, and it does require a very accurate and precise function to compute the log of the binomial $\logBinomial(n, j) = \log\left(\binom{n}{j}\right)$, which isn't trivial.
 An approach one might consider is to use
\begin{align}
\label{eq:logBinomial}
\log \binom{n}{j} & = \log\Gamma(n+1) -  \log\Gamma(j+1) -  \log\Gamma(n-j+1)
\end{align}
Just computing $\log\Gamma(m)$ accurately can be an expensive operation, and subtractive cancellation is an issue for $j \ll n$.
Computing $S_n(x)$ via \eqnref{sm_2alpha} and \eqnref{Ajnlog} appear equivalent, and would be if both were computed with infinite precision arithmetic.
The latter has fewer underflow/overflow issues but bigger accuracy issues.

%%%%%%%%%%%%%%%%%%%%%%%%%%%%%%%%%%%%%%%%%%%%%%%%%%%%%
%%%%%%%%%%%%%%%%%%%%%%%%%%%%%%%%%%%%%%%%%%%%%%%%%%%%%
%%%%%%%%%%%%%%%%%%%%%%%%%%%%%%%%%%%%%%%%%%%%%%%%%%%%%

%%%%%%%%%%%%%%%%%%%%%%%%%%%%%%%%%%%%%%%%%%%%%%%%%%%
%%%%%%%%%%%%%%%%%%%%%%%%%%%%%%%%%%%%%%%%%%%%%%%%%%%
%%%%%%%%%%%%%%%%%%%%%%%%%%%%%%%%%%%%%%%%%%%%%%%%%%%

\subsection{Evaluation near x=0}

If $x$ is very close to 0, $S_n(x)$ is very close to 1.
\eqnref{sm_2alpha} calculates $S_n(x)$ using almost all the $n$ terms, which allows many opportunities for rounding errors to accumulate.
Since the $\CDF$ is calculated as $1-S_n(x)$ the $\CDF$ values can be quite noisy.
However there is a variant formula \cite{zbMATH03107305, dwass1959},
(referred to as DwassD in \cite{JSSv019i06}  but already present in Smirnov's 1944 paper),
equivalent due to one of Abel's combinatorial identities \cite{abel1826, riordan1968combinatorial}, suitable for use whenever $nx$ is small:
\begin{align}
\label{eq:Sn_upper}
Pr(D_n^+ \leq x) & =  x\sum_{j=\lceil n(1-x)\rceil}^{n} \binom{n}{j}  \left(x+\frac{j}{n}\right)^{j-1} \left(1-x-\frac{j}{n}\right)^{n-j}  \\
& = x\sum_{j=\lceil n(1-x)\rceil}^{n}  A_j(n, x)  \\
\shortintertext{which simplifies substantially  if $nx<=1$}
\label{eq:Sn_upper1}
Pr(D_n^+ \leq x) & =  x (1+x)^{n-1}  \quad\quad\text{for\space} 0 \leq x\leq  \frac{1}{n}
\end{align}

The terms in the sum in \eqnref{Sn_upper} are alternating in sign and may involve numbers much bigger than 1 
which when added mostly cancel each other out.
(See \figref{Ajn}(c).)  Substantial rounding error may have accumulated during the process, leaving the sum without much accuracy.
For a fixed integer $k$, and any $0 \leq j < k$,
\begin{align}
A_{n-j}(n, \frac{k}{n}) &  \underset{n\rightarrow \infty}{\asymp}  (-1)^j \frac{(k-j)^{j}}{j!}e^{k-j}\\
\shortintertext{so that}
Pr(D_n^+ \leq \frac{k}{n})  & \underset{n\rightarrow \infty}{\asymp}  \frac{k}{n}  \left(e^k - (k-1)e^{k-1} + \frac{(k-2)^2}{2!} e^{k-2} -  \frac{(k-3)^3}{3!} e^{k-3} \dots\right)
\end{align}
though the convergence is slow.
Computing $S_n(\frac{k}{n})$ in this manner loses about $1.6*(k-1)$ bits.
Brown and Harvey \cite{JSSv026i03} analyzed the extra internal digits of precision needed to compute using \formularef{Sn_upper}.
They found that the number of extra digits needed grows like $O(\sqrt{n})$ (keeping the $p_{SF}$ fixed and evaluating at $x_n=S_n^{-1}(p_{SF})$).
In particular, evaluating \eqnref{Sn_upper} for $n$ up to 800 for $p_{SF}=0.01$  (i.e. $x\approx 1.5/\sqrt{n}$) loses about 27 digits; for $p_{SF}=0.7$ ($x\approx 0.4/\sqrt{n}$), it loses about 8 digits.
While this this formula undoubtedly is faster for $x \leq 0.5$ due to the smaller number of terms,
it is not suitable for use in a fixed precision environment unless $n*x$ is small.

%%%%%%%%%%%%%%%%%%%%%%%%%%%%%%%%%%%%%%%%%%%%%%%%%%%
%%%%%%%%%%%%%%%%%%%%%%%%%%%%%%%%%%%%%%%%%%%%%%%%%%%
%%%%%%%%%%%%%%%%%%%%%%%%%%%%%%%%%%%%%%%%%%%%%%%%%%%

%%%%%%%%%%%%%%%%%%%%%%%%%%%%%%%%%%%%%%%%%%%%%%%%%%%
%%%%%%%%%%%%%%%%%%%%%%%%%%%%%%%%%%%%%%%%%%%%%%%%%%%
%%%%%%%%%%%%%%%%%%%%%%%%%%%%%%%%%%%%%%%%%%%%%%%%%%%

\subsection{Evaluation for very large n}

If $2nx^2 > \log(2^{1023 + 52}) \approx 745$,
then $S_n(x) \leq \exp(-2nx^2)$ and will underflow 64 bit floats.
(If $2nx^2 > \log(2^{1023})$ but less than 745, the result will be denormalized but still representable.) 
  I.e. If $x\sqrt{n} \gtrapprox 19.3$, then the $\SF$ is 0, the $\CDF$ 1, with no additional computation needed.  Hence for very large $n$, the only $x$ of interest satisfy $0 <x \lessapprox  \frac{20}{\sqrt{n}}$.  The asymptotics \eqnsref{SmirnovAsymptote}{MaagAsymptote}
are appropriate for a fixed $x$, not for $x$ dependent on $n$.
(E.g. $S_n(\frac{1}{n}) \approx 1-\frac{e}{n+1}$ whereas \eqnref{MaagAsymptote} predicts $e^{\frac{-37}{18n}}$, not at all similar.)
For $x \lessapprox \frac{1}{\sqrt{n}}$, the majority of the terms in \formularef{sm_2alpha} could conceivably contribute to the total, as seen in \figref{Ajn}.
As $x$ increases the magnitude of the relevant terms becomes
more peaked, centered at $j = \frac{n - nx}{2}$.
However, even as $x\uparrow\frac{20}{\sqrt{n}}$, some 10\% of the terms are potentially relevant for $n=40,000$ and up.

As $n\to\infty$ the survival probabilities for $x=\frac{k}{n}$ computed using the asymptote \eqnref{MaagAsymptote} agree with $S_n(x)$ in about $\log_2(n)$ bits, but mainly because $S_n(x) \uparrow 1$ so all the bits eventually become 1!
For $\CDF$ probabilities computed this way, there might be only $8-13$ bits (out of 53) in agreement, and it doesn't improve with increasing $n$.

%%%%%%%%%%%%%%%%%%%%%%%%%%%%%%%%%%%%%%%%%%%%%%%%%%%
%%%%%%%%%%%%%%%%%%%%%%%%%%%%%%%%%%%%%%%%%%%%%%%%%%%
%%%%%%%%%%%%%%%%%%%%%%%%%%%%%%%%%%%%%%%%%%%%%%%%%%%

%%%%%%%%%%%%%%%%%%%%%%%%%%%%%%%%%%%%%%%%%%%%%%%%%%%
%%%%%%%%%%%%%%%%%%%%%%%%%%%%%%%%%%%%%%%%%%%%%%%%%%%
%%%%%%%%%%%%%%%%%%%%%%%%%%%%%%%%%%%%%%%%%%%%%%%%%%%

\subsection{Evaluation of the PDF}

The $\PDF$ for $x \neq 0$ can be evaluated after realizing
%with the help of \eqnref{dAjn}  so that
\begin{align}
\label{eq:dAjn}
\frac{d A_j(n, x) }{dx} & = \left(\frac{j-1}{x+j/n} - \frac{n-j}{1-x-j/n}  \right) A_j(n, x)   \qquad\text{for $x \neq \frac{n-j}{n}$}\\
\shortintertext{so that}
\PDF(x) & = \frac{d\, Pr(D_n^+ \leq x)}{dx}\\
\label{eq:dSnDeriv}
 & =  -x  \sum_{j=0}^{\lfloor n(1-x)\rfloor} \left(\frac{1}{x} + \frac{j-1}{x+j/n} - \frac{n-j}{1-x-j/n}  \right) A_j(n, x)\\
\label{eq:dSn}
& =  \frac{-S_n(x)}{x} -x  \sum_{j=0}^{\lfloor n(1-x)\rfloor} \left(\frac{j-1}{x+j/n} - \frac{n-j}{1-x-j/n}  \right) A_j(n, x)\\
\label{eq:dSnupper}
 & = \frac{1-S_n(x)}{x} +  x  \sum_{j=\lceil n(1-x)\rceil}^{n} \left(\frac{j-1}{x+j/n} - \frac{n-j}{1-x-j/n}  \right) A_j(n, x)
\\
\label{eq:dSnupperDeriv}
& = x  \sum_{j=\lceil n(1-x)\rceil}^{n} \left(\frac{1}{x}  + \frac{j-1}{x+j/n} - \frac{n-j}{1-x-j/n}  \right) A_j(n, x)
\end{align}
In the computation of the $\CDF$ via \eqnref{sm_1}, the top limit was replaced by $\lceil n(1-x)\rceil-1$, as $A_j(n, x)$ has a zero of order $n-j$ at $x=\frac{n-j}{n}$. 
The same can be done here with the $\PDF$ for all $x\neq\frac{1}{n}$, as the \ordinalnum{1} derivative of $A_j(n, x)$ vanishes at $x=\frac{n-j}{n}$ for $j < n-1$.
From \eqnref{S2}, we see that $\lim_{x\uparrow 0.5}S_2'(x) = -2$ but $\lim_{x\downarrow 0.5}S_2'(x) = -1$, hence there is a discontinuity at $x=0.5$.   In general there is a discontinuity at $x=\frac{1}{n}$,  just discernible in \figref{smirnov_and_sqrtn} and clearly visible in \figref{Sn_prime}.

This formula also shows that the $\PDF$ can be evaluated either on its own with a slight modification to \eqnref{sm_1}, or jointly with the $\SF$/$\CDF$ and only a little incremental effort, as most of the computational cost in \eqnref{dSn} is in computing the powers and/or binomial coefficients of $A_j(n, x)$.
The coefficient of $A_j(n, x)$ in \eqnref{dSn}  are all positive which can be helpful for the summation. But we have observed some subtractive cancellation when incorporating the $\frac{-S_n(x)}{x}$ term.
The coefficients of \eqnref{dSnDeriv} and \eqnref{dSnupperDeriv} have mixed sign, but the subtractive cancellation inside the parentheses occurs when the terms are smaller in magnitude, leading to a smaller absolute error.

%%%%%%%%%%%%%%%%%%%%%%%%%%%%%%%%%%%%%%%%%%%%%%%%%%%
%%%%%%%%%%%%%%%%%%%%%%%%%%%%%%%%%%%%%%%%%%%%%%%%%%%
%%%%%%%%%%%%%%%%%%%%%%%%%%%%%%%%%%%%%%%%%%%%%%%%%%%

%%%%%%%%%%%%%%%%%%%%%%%%%%%%%%%%%%%%%%%%%%%%%%%%%%%
%%%%%%%%%%%%%%%%%%%%%%%%%%%%%%%%%%%%%%%%%%%%%%%%%%%
%%%%%%%%%%%%%%%%%%%%%%%%%%%%%%%%%%%%%%%%%%%%%%%%%%%

\subsection{The SciPy implementation}

The implementation in SciPy uses \formularef{sm_2alpha} for $n < 1013$ and \formulasref{Ajnlog}{logBinomial} for $n > 1013$.

\begin{itemize}
\item
For small $n$ (less than a few hundred), there is a loss of accuracy due to evaluating $(x+j/n)^m$ as $\pow(x+j/n, m)$.

\item For large enough $n$ ($n > 1013$) there is some loss of accuracy for small $x$ due to evaluation of $(1-x-j/n)$ as $(1-(x+j/n))$, and using $\log$ always instead of using $\logonep$ when the latter would be more appropriate.

\item An even bigger problem is that $\logBinomial(n, j)$  is calculated as $\loggamma(n+1) - \loggamma(n-j+1) - \loggamma(j+1)$.
This suffers from a substantial loss of accuracy for small values of $j$, due to subtractive cancellation.
The binomial coefficient can also be expressed in terms of the Beta function:
\begin{align}
  B(x, y) & = \frac{\Gamma(x)\Gamma(y)}{\Gamma(x, y)}
// \binom{n}{j} & = \frac{1}{B(j+1, n-j+1)} \frac{1}{n+1}
\end{align}
and SciPy does have a function $\lbeta$ which computes $\log(B(a, b))$.
This function has lower error for small $(n, j)$ than computing the logs of Gamma functions.
However if $n+j> 170$, the implementation of $\lbeta$  switches algorithms,
 reverting to computing logs of Gamma functions,
 immediately losing up to \textasciitilde10 bits of precision, so this does not help for $n > 1000$.

\item For the mid-range $n$, the computation may return values much smaller than the true probabilities.
One would expect that $S_n(x)$ is a decreasing function of $n$ (a bigger sample should imply smaller gaps in the EDF).
But a jump was discovered between $n=1012, 1013$, in that the returned
values of $\smirnov(1013, x)$ were {\em higher}\/ than those of $\smirnov(1012, x)$.
(E.g.  For $x=0.45$, $\smirnov(1013, x) > 10^{10} \smirnov(1012, x)$.)
This was due to underflow in  $\pow(1-x-j/n, m)$, which resulted in most terms in the summation being calculated as 0.
While the closer $n$ is to 1012, the bigger the effect, the differences have been observed for $n$ as low as 400.
A more accurate computation revealed that the SciPy result could be off by a factor of \num{e55}!
Due to the difficulties in evaluating $\pow(x+j/n, m)$, the
calculation could have benefitted from switching over to using $\log$ at a much smaller $n$.

\item SciPy doesn't provide a separate function to compute the $\PDF$.
Instead it numerically differentiates the $\CDF$: $1-S_n(x)$.
This involves multiple evaluations of $\smirnov(n, x)$, and is a little unstable near $x=\frac{j}{n}$, especially at $x=\frac{j}{n}$.
For $n=2$, the $\PDF$ should jump from 2 to 1 at $x=0.5$.
The computed $\PDF$ was observed to spike up and takes values as high as $2.08$ as $x\uparrow0.5$;
as low as $0.91$ as $x\downarrow0.5$; and the value returned at $x=0.5$ was $\approx 1.5$.

\end{itemize}

%%%%%%%%%%%%%%%%%%%%%%%%%%%%%%%%%%%%%%%%%%%%%%%%%%%
%%%%%%%%%%%%%%%%%%%%%%%%%%%%%%%%%%%%%%%%%%%%%%%%%%%
%%%%%%%%%%%%%%%%%%%%%%%%%%%%%%%%%%%%%%%%%%%%%%%%%%%

    %!TEX root = ./ms.tex

%%%%%%%%%%%%%%%%%%%%%%%%%%%%%%%%%%%%%%%%%%%%%%%%%%%
%%%%%%%%%%%%%%%%%%%%%%%%%%%%%%%%%%%%%%%%%%%%%%%%%%%
%%%%%%%%%%%%%%%%%%%%%%%%%%%%%%%%%%%%%%%%%%%%%%%%%%%

\section[Algorithms for computing Smirnov SF, CDF and PDF]{Algorithms for computing Smirnov $\SF$, $\CDF$ and $\PDF$ }
\label{sec:smirnov_prop}

The \SF is a sum of many terms, each term is a product of expressions that underflow/overflow or are hard to compute accurately.  
Here we propose algorithms which address the issues discovered with computing the one-sided Smirnov probability.

\subsection{Smirnov}

First we provide a somewhat generic algorithm for computing $S_n(x)$ and subsequently provide the details of computing the individual terms.
Let $D_j(n, x)$ denote the term inside the $\PDF$ summation \eqnref{dSnDeriv}:
\begin{align*}
D_j(n, x) &= \left(\frac{1}{x}  + \frac{j-1}{x+j/n} - \frac{n-j}{1-x-j/n}  \right) A_j(n, x)
\end{align*}

\setcounter{pvmalgorithm}{\value{algorithm}}
\begin{pvmalgorithm}[smirnov]
\label{alg:smirnov_alg}
Compute the Smirnov $\SF$, $\CDF$ and $\PDF$ for integer $n$ and real $x$.
\end{pvmalgorithm}
\setcounter{algorithm}{\value{pvmalgorithm}}

\code{function [SF, CDF, PDF] = smirnov(n:int, x:real)}
\noindent
\begin{enumerate}[start=1,label={{\bf Step \arabic*}},ref={Step \arabic*}]

\item Handle some special cases:
\label{smalg_special}
\begin{enumerate}[label={(\roman*)}]
\item If $x<0$ or $x>1$, return $[1, 0, 0]$, $[0, 1, 0]$ respectively.
\item If $n=1$, return $[1-x, x, 1]$.
\item If $x=0$ or $x=1$, return $[1, 0, 1]$ or $[0, 1, 0]$ respectively.
\item If $2*n*x^2 > 745$,  return $[1, 0, 0]$.
\label{smalg_bignx}
\end{enumerate}

\item
\label{smalg_alpha}
Set $k \leftarrow \lfloor{n*x}\rfloor$, $\alpha \leftarrow n*x - k$, so that
$x = \frac{k+\alpha}{n}$ with $0 \leq \alpha < 1$.

\item If $nx$ is ``small enough", select the alternate Smirnov/Dwass summation \eqnref{Sn_upper}, and \eqnref{dSnupperDeriv} to use for the computation and set
$N \leftarrow k$.
Otherwise select \eqnref{sm_2alpha} and \eqnref{dSnDeriv} and set
$N\leftarrow  n-k-1$.
\label{smalg_upper}

\item If $N$ is ``very large", set
\label{smalg_verylargen}
    \begin{align*}
     E & \leftarrow  -(6*n*x+1)^2/(18*n) \\
{SF} & \leftarrow  \code{exp(E)}\\
      {CDF} & \leftarrow    \code{-expm1(E)} \\
      {PDF} & \leftarrow  {(6*n*x + 1)* 2*{SF}/3}
   \end{align*}
    Return [${SF}, {CDF}, {PDF}$].

\item
\label{smalg_j0}
Initialize the loop. Set
\begin{align*}
C, E_C & \leftarrow 1, 0 \\
A_0 & = \begin{cases}
    \mathrlap{(1+x)^{n-1} }
    % \hphantom{\dfrac{1-(n-1)x}{(1-x)} \, A_0}
    \hphantom{\dfrac{1+nx}{x(1+x)} \, A_0 }
    & \text{for Smirnov/Dwass} \\
    \dfrac{(1-x)^n}{x} & \text{otherwise}
    \end{cases}\\
D_0 & = \begin{cases}
    \dfrac{1+nx}{x(1+x)} \, A_0 & \text{for Smirnov/Dwass} \\
    \dfrac{-n}{1-x} \, A_0& \text{otherwise}
    \end{cases}
\end{align*}

\item
Loop  over $i \leftarrow 1, 2, \dots,  N$:
\label{smalg_oneterm}
\label{smalg_pdf}
Set
\begin{align*}
C & \leftarrow   C * (n-i+1)/i  %& \text{if $i > 0$.}
  \notag \\
C, \; E_C & \leftarrow
    \begin{cases}
     \mathrlap{C /2^{512},} \hphantom{C * 2^{512},} \; E_C - 512  &\text{\qquad if $C > 2^{512}$} \\
     C * 2^{512},  \;E_C + 512  &\text{\qquad if $C  < 1$} \\
     \mathrlap{C,} \hphantom{C * 2^{512},}  \;  E_C  &\text{\qquad otherwise}
    \end{cases}
\\
%\shortintertext{For all $i$,set}
    j & \leftarrow
        \begin{cases}
        n-i    %  \mathrlap{n-i}          \hphantom{ \dfrac{-n(n-i+(n-1)k +(n-1)\alpha)} {(i+k+\alpha)(n-i-k-\alpha)} }
         &  \text{for Smirnov/Dwass} \\
         i & \text{otherwise}
        \end{cases}  \notag\\
    A_i & \leftarrow  A_{j}(n, x)  \notag\\
    D_i & \leftarrow  D_{j}(n, x)  \notag
\end{align*}

\item Sum the terms: $A \leftarrow  \sum A_i$ and $D \leftarrow  \sum D_i$.
\label{smalg_sum}

\item Set [${SF}, {CDF}, {PDF}] \leftarrow  [x*A, \; 1-x*A, \;  -x*D]$.  \\
(Set [${SF}, {CDF}, {PDF}] \leftarrow  [1-x*A,\;  x*A, \; x*D]$   if using the Smirnov/Dwass alternate summation.)

\item
\label{smalg_clip}
Clip ${SF}, {CDF}$ to lie in the interval $[0, 1]$ and ${PDF}$ to lie in $[0, \infty)$.
Return [${SF}, {CDF}, {PDF}$].

\end{enumerate}

%%%%%%%%%%%%%%%%%%%%%%%%%%%%%%%%%%%%%%%%%%%%%%%%%%%
\subsubsection{Remarks} % Commentary

Computations performed using 64-bit floats often were the cause of accuracy loss.
In a few places during the computation using the double length operations
\cite{Dekker:1971:FTE:2716631.2717032}, \cite{Knuth:1997:ACP:270146}, \cite{DoubleDoubleBailey} were better keeping the errors smaller.
As these operations are usually not supported in hardware, hence consume more computation, it is not always practical to use them.
We'll point out several places where selective use of these operations made a big difference.

\begin {itemize}
\item Specifying  the API to return both the $\CDF$ and $\SF$ probabilities enables
both
to retain as much accuracy as had been computed, without the need for a separate function.
The $\PDF$ can be computed with just a little extra work, and it is often needed at the same time as the probability functions.

\item If $2*n*x^2 > 745$, (\stepref{smalg_special}~\ref{smalg_bignx}), then $S_n(x) < \exp(-2nx^2)$ and $\log(2^{-1075}) = e^{-745.13\dots}$ underflows 64-bit floats, so there is no need to do any computation. (For $2^{-1075} \leq 2*n*x^2 \leq 2^{-1052}$ the result will be denormalized but still representable to some degree.)

\item Computing $\alpha$ (\stepref{smalg_alpha}) via C's $\modf(n*x)$ is a significant source of subsequent increased error if performed in 64-bit arithmetic.
The problem is that even though $x\in\mathbb{F}$, the product $nx$ is not necessarily in $\mathbb{F}$.
However it is the exact sum of two floating point values.  Let $\epsilon=2^{-53}$ and assume $n < 1/\epsilon \approx \num{d16}$.
Then we can write
\begin{align}
nx &= U+V
\shortintertext{with $U, V\in \mathbb{F}$ and  $|V| \leq \epsilon U \ll 1$.
The integer part of $nx$ is {\em{almost}}\/ the integer part of $U$.
Let $k=\lfloor U\rfloor$.  Then}
nx &= k  + ((U-k)+V)
\\ &= k + (U_1 + V_1)
\end{align}
with $U_1, V_1\in \mathbb{F}$ and $|V_1| \leq \epsilon|U_1|$ ($V_1$ might be 0).
\begin{enumerate}
\item If $0 \leq U_1 < 1$, then $k=\lfloor nx\rfloor$ and $\alpha=(U_1 + V_1)$.
\item If $U_1>1$ or $U_1=1$ and $V_1\geq0$, set $k\leftarrow k+1$ and $\alpha \leftarrow (U_1-1.0)+ V_1$.
\item If $U_1=1$ and $V_1<0$, then $\alpha$ is a tiny bit smaller than 1.0, but not representable as a single floating point value.
Set $k\leftarrow k+1$ and $\alpha \leftarrow 0.0$.
\item If $U_1<0$ then $U_1+1$ is not representable as a single floating point value.
Set
$\alpha \leftarrow 0.0$.
\end{enumerate}
For values of $x$ close to $\frac{j}{n}$, this computation needs some special care, as otherwise the limits of the summation can be off by 1.
For $x$ close to $\frac{1}{n}$, the impact on the $\PDF$ is to add $\pm1$ to the correct value.

\item Because use of the Smirnov/Dwass alternate summation in \stepref{smalg_upper} involves summing large alternately signed quantities which mostly cancel,
having to sum more than a few terms can lead to a loss in precision.
Restricting to $x<= \frac{1}{n}$ is a very conservative criterion to use for this, with no accuracy loss incurred.
One other point of consideration here is the behaviour of {\code{smirnov(n, x)}} on either side of the changeover $x$-value, and whether there is a big discontinuity in the values computed. It may also interact with the
calculation of the inverse SF function: $S_n^{-1}(p_{SF})$.  The inverse function may be relying on particular behaviour/properties of $S_n(x)$
 (such as $S_n(x)\leq x(1+x)^{n-1} \iff x\leq \frac{1}{n}$, or that $\lim_{x\uparrow x_c} S_n(x) = \lim_{x\downarrow x_c} S_n(x)$ for any $x_c$)
 and it is expecting the computed values to also respect those properties.
Changing over at $x=\frac{1}{n}$ is easy to analyze and isolate, as the formula \eqnref{Sn_upper} is easy to calculate, and handled separately in $\smirnovi$.  Other change over points may require more analysis.
Use of higher precision arithmetic would allow for a higher cutoff to be used and still have a low error.

\item For very large $n$, the computational cost may become prohibitive, or the accumulated rounding errors not small enough.
In those situations an approximation may be in order.
\stepref{smalg_verylargen} is quick to compute, and generates probabilities that have a few bits correct.
For some applications that may be sufficient.
Computing and summing 100 million terms is a noticeable calculation, but is certainly doable.
Only invoking this step for $n > \num{e12}$ may be one approach which allows most legitimate uses to proceed and avoids having to determine accuracy guarantees for such large $n$.

\item \stepref{smalg_j0} is an initialization step, mainly for the binomial $C$, and the simplified formulae for $A_0, D_0$. $E_C$ is the $\log_2$ of a scaling factor, without which $C$ would overflow.

\item The computation of a single $A_j(n, x)$ term (\stepref{smalg_oneterm}) is covered below.  The computation of the $\PDF$ term
uses a safe formulation of \eqnref{dSnDeriv} and \eqnref{dSnupperDeriv}.
The binomial coefficient $C$ is normalized at each step to be between $1$ and $2^{512}$.

\item
For \stepref{smalg_sum}, the summation in \eqnref{sm_2alpha}  has up to $n$ terms, but all the terms have the same sign, hence the condition number for the sum is 1.
Using a Neumaier/Kahan \cite{Kahan:1965:PRR:363707.363723,ZAMM:ZAMM19740540106} summation method appears sufficient.
The summation in \eqnref{dSn} also has up $n$ terms of the same sign, but using it is only postponing the cancellation that will occur when the $\frac{S_n(x)}{x}$ contribution is added in.  \eqnref{dSnDeriv}  has mixed signs.
The alternate summations \eqnref{sm_2alpha} and \eqnref{dSnupperDeriv} have fewer terms, but the terms alternate in sign, so have a higher condition number.
Use a compensated summation algorithm (\eg Rump et al's \cite{Rump:2008:AFS:1461600.1461604} accurate summation)
if storage is available for all the individual terms and the alternate summation is being used for more than a few terms, or Neumaier/Kahan otherwise.

\item Clipping in \stepref{smalg_clip} is there so that any rounding issues which resulted in slightly negative probabilities are not propagated to the user.

\item 
The calculation can be terminated early if the $A_i$ term becomes too small to affect the outcome.  Due to the plateau nature of the  $A_i$, this is usually only possible near the very end of the loop, only saving a few iterations.  Depending on the value of $\alpha$, the term for $j=\lfloor n-nx\rfloor$ may contribute significantly to the derivative computation (as $1-x-j/n=1-\alpha$ may be very close to 0), hence this was found not to be a worthy optimization in most situations.

\end{itemize}

%%%%%%%%%%%%%%%%%%%%%%%%%%%%%%%%%%%%%%%%%%%%%%%%%%%
%%%%%%%%%%%%%%%%%%%%%%%%%%%%%%%%%%%%%%%%%%%%%%%%%%%
%%%%%%%%%%%%%%%%%%%%%%%%%%%%%%%%%%%%%%%%%%%%%%%%%%%
%%%%%%%%%%%%%%%%%%%%%%%%%%%%%%%%%%%%%%%%%%%%%%%%%%%

\subsection[Error Analysis of Computations of A_j(n, x)]{Error Analysis of Computations of $A_j(n, x)$}

This section is a fairly technical one, analyzing the maximum errors involved when computing using either \code{pow} or \code{log/exp},
in order to guide selection of a particular algorithm to compute $A_j(n, x)$.

Let $\delta_{pow}$ ($\delta_{log}$, $\delta_{exp}$) be the maximum absolute relative error in a call to \code{pow(x, m)} (\code{log(x)}, \code{exp(x)} respectively), $\delta_{\times}$ the maximum relative error in a multiplication/division, $\delta_{+}$ the maximum relative error in an addition/subtraction.
I.e. 
\begin{align}
 \code{fl}(x \operatorname{op} y)&  = (x \operatorname{op} y)(1+\delta)
\end{align}
%$\code{fl}(x \operatorname{op} y) = (x \operatorname{op} y)(1+\delta)$ 
with $|\delta | \leq \delta_{+}$ for $\operatorname{op} = + \text{ or } -$. $\delta_{+}$ is usually the unit roundoff, which is $2^{-53}$ for 64 bit doubles and $\delta_{x}$ is usually equal to $\delta_{+}$. They will be left abstract to allow analysis for other representations.  We will abbreviate all \ordinalnum{2} and higher order terms such as $O(\delta_{\times}^2, \delta_{pow}^2, \delta_{\times}\delta_{pow})$ with $O(\delta^2, \dots)$.

\begin{theorem}[Error in Computation of $A_j(n, x)$]
\label{theorem:AjnError}
Set
\begin{align*}
C_{n,j} & \leftarrow \code{(n/1) * ((n-1)/2) * \dots{} * ((n-j+1)/j)} \\
\shortintertext{and}
\hat{A_j} & \leftarrow C_{n,j} *  \code{pow((j+k+$\alpha$)/n,  j-1)} *  \code{pow((n-j-k-$\alpha$)/n,  n-j)} \\
\shortintertext{If none of the terms underflow, overflow or become denormalized, then}
\hat{A_j} & =  A_j(n, x) \cdot (1 + \delta_{A_j})\\
\shortintertext{with}
|\delta_{A_j}| & \leq  (n+2j)\delta_{\times} + (n-1)\delta_{+} + 2\delta_{pow} + O(\delta^2, \dots).
\end{align*}
\end{theorem}

This immediately implies a bound on the error computing $S_n(x)$.
\begin{theorem}[Error in Computation of $S_n(x)$]
\label{theorem:SnError}
Let $A_j$ be an estimate of $A_j(n, x)$ with relative error $\delta_{A_j}$, for $j=0, \dots N=\lfloor n(1-x)\rfloor$.
Define
\begin{align*}
 \gamma_{n} & \triangleq \frac{n\delta_{+}}{1-n\delta_{+}} \\
\shortintertext{Set}
S_{n} & \leftarrow \code{$A_0$ + $A_1$ + \dots{} + $A_N$} \\
\shortintertext{Then}
S_n & =  S_n(x) \cdot (1 + \delta_{S_n}) \\
\shortintertext{with}
|\delta_{S_n}| & \leq   \gamma_{N} \max(|\delta_{A_j}|) + O(\delta^2, \dots) \\
\shortintertext{Using the bounds of \theoremref{AjnError}, this can be bounded as}
& \leq  \gamma_{N} (3n\delta_{\times} + (n-1)\delta_{+} + 2\delta_{pow}) + O(\delta^2, \dots)
\end{align*}
\end{theorem}

\begin{proof}
The summation bound is standard (E.g. \cite{Higham1993}) and the terms all have the same sign, so the condition number is 1.
\end{proof}

We will break the proof of \theoremref{AjnError} into several lemmas.

\begin{lemma}[Errors in Computations involving standard functions]
Let $x\in \mathbb{R}$, $x_f \in \mathbb{F}$ an approximation to $x$, $m\in\mathbb{Z}$.
\label{lemma:stdfuncerrors}
Then
\begin{enumerate}[label={\emph{\alph*})}]
\item
$x^m  = \code{pow($x_f$, m)} \cdot (1+\lambda \delta_{pow})\cdot(1+\frac{x-x_f}{x_f})^m \text{\quad for some $-1 \leq \lambda \leq 1$} $\\
$ \hphantom{x^m} \approx  \code{pow($x_f$, m)} \cdot (1+\lambda \delta_{pow}+ m\frac{x-x_f}{x})$

\item $\log_e(x) = \code{log($x_f$)} \cdot (1+ \lambda \delta_{log}) + \log_e(1+\frac{x-x_f}{x_f})$ \text{\quad for some $-1 \leq \lambda \leq 1$} \\
 $\hphantom{\log_e(x)} \approx \code{log($x_f$)} \cdot (1+ \lambda \delta_{log}) + \frac{x-x_f}{x}$  \\
 $\hphantom{\log_e(x)} \approx \code{log($x_f$)} \cdot \left(1+ \lambda \delta_{log} + \frac{(x-x_f)/x}{\log(x_f)}\right)$

\item $e^x = \code{exp($x_f$)}  (1+ \lambda \delta_{exp})  \cdot e^{x-x_f}$ \text{\quad for some $-1 \leq \lambda \leq 1$} \\
$ \hphantom{e^x} \approx  \code{exp($x_f$)} \cdot \left(1+ \lambda \delta_{exp} + x\left(\frac{x-x_f}{x}\right)\right)$
\end{enumerate}
\end{lemma}

\begin{proof}
The terms involving $\lambda$ correspond to errors in the output of the implementations of the standard functions, the other terms correspond to the errors in the inputs to these functions.
$\frac{x-x_f}{x}$ is the relative error in the approximation of $x$ by $x_f$.
\end{proof}

\begin{table}[!htb]
\centering
\label{tab:errorexpansion}
\begin{tabular}{ c | c }
\hline
$g(x)$ & $\frac{xg'(x)}{g(x)}$ \\
\hline
$x^m$ & $m$ \\
$\exp x$ & $x$ \\
$\log x$ & $\frac{1}{\log x}$ \\
$\log(1+x)$ & $ \frac{x}{(1+x)\log(1+x)}$ \\
\hline
\end{tabular}
\caption{Relative error multipliers for some functions: $g(x(1+\delta)) \approx g(x)(1+\frac{xg'}{g}\delta)$.}
\end{table}

\begin{lemma}[Error in Computation of Binomial Coefficient]
$C_{n,j}$ approximates the binomial coefficient $\binom{n}{j}$
\[ C_{n,j} = \binom{n}{j}  \cdot (1+\delta_C)\]
and the relative error satisfies
\[|\delta_C| \leq  (2j-1)\delta_{\times} + O(\delta_{\times}^2).\]
\end{lemma}
\begin{proof}
There are $j$ divisions and $j-1$ multiplications.
\end{proof}

 \begin{lemma}[Error in Computation of Powers]
\label{lemma:powererrors}
Set P $\leftarrow$ \code{pow((a+b)/n, m)}, with $a, b>0$.  Then
\begin{enumerate}[label={\emph{\alph*})}]
\item
\label{singlepow}
\begin{align*}
P & = \left(\frac{a+b}{n}\right)^{m} (1+\delta_P)\\
\shortintertext{with}
|\delta_P| & \leq  m(\delta_{\times} + \delta_{+}) +  \delta_{pow} + O(\delta^2, \dots).
\end{align*}

\item
For a fixed $d\in\mathbb{Z}^{+}$ write $m$ in base $d$: $m = r_0 + r_1d + r_2d^2 \dots$ with $0\leq r_i < d$.  Set
\begin{align*}
Q_0 & \leftarrow \code{(a+b)/n}  \\
\shortintertext{For $0 \leq j\leq \log_{d}(m)$ set}
P_j & \leftarrow \code{pow($Q_{j}$, $r_{j}$)}  \\
Q_{j+1} & \leftarrow \code{pow($Q_{j}$, d)}  \\
\shortintertext{and}
\hat{P} & \leftarrow \code{$P_0 * P_1 * P_2 * \dots$}\\
\shortintertext{Then}
\hat{P} & = \left(\frac{a+b}{n}\right)^{m} (1+\delta_{\hat{P}})\\
\shortintertext{with}
|\delta_{\hat{P}}| & \leq  m(\delta_{\times} + \delta_{+}) + \left(\sum_{j=1}^{\infty} \lfloor{\frac{m}{d^j}}\rfloor + \lceil{\log_{d}m}\rceil\right)\delta_{pow}
   + \lfloor{\log_{d}m}\rfloor\delta_{\times} + O(\delta^2, \dots)\\
& \leq  m(\delta_{\times} + \delta_{+}) + \left(\frac{m(d+1)}{d^2} + \lceil{\log_{d}m}\rceil\right)\delta_{pow}
   + \lfloor{\log_{d}m}\rfloor\delta_{\times} + O(\delta^2, \dots).
\end{align*}
\end{enumerate}
\end{lemma}
\begin{proof}
\begin{enumerate}[label={\emph{\alph*})}]
\item Apply \lemmaref{stdfuncerrors}, with $x=\frac{a+b}{n}$.
\item Apply part~\ref{singlepow} recursively.
\end{enumerate}
\end{proof}

\begin{proof}[Proof of \theoremref{AjnError}]
The binomial coefficient computation contributes $(2j-1)\delta_{\times}$, the two powers contribute
$(j-1)*(\delta_{\times} + \delta_{+}) +  \delta_{pow}$ and $(n-j)*(\delta_{\times} + \delta_{+}) +  \delta_{pow}$, respectively, and there are two additional multiplications.
\end{proof}

%%%%%%%%%%%%%%%%%%%%%
%%%%%%%%%%%%%%%%%%%%%

In unscaled, the binomial coefficient will overflow for large $n$, and the output of \code{pow} underflow.
Next we analyze the errors when using \code{log/exp}.

\begin{theorem}[Error in Computation of $A_j(n, x)$ using log/exp]
\label{theorem:AjnLogError}
Set
\begin{align*}
L_j & \leftarrow  \code{log(n/1.0) + log((n-1)/2.0) + \dots + log((n-j+1)/j)}\\
{T_1} & \leftarrow \code{(j-1)*log((j+k+$\alpha$)/n)}\\
{T_2} & \leftarrow \code{(n-j)*log((n-j-k-$\alpha$)/n)}\\
T & \leftarrow  \code{$L_j$ + {$T_1$} + {$T_2$}}\\
A_{j,log} & \leftarrow  \code{exp(T)}\\
\shortintertext{If $A_{j,log}$ has not underflowed, overflowed or become denormalized, then}
A_{j,log} & = A_j(n, x) \cdot (1 + \delta_{A_{j,log}})  \\
\shortintertext{with}
\begin{split}
|\delta_{A_{j,log}}| & \leq  \delta_{exp} + |T| (3 \delta_{log} + \gamma_{j-1}R(L_j) + 2\delta_{\times}) \\
& \quad + (n-1)(\delta_{\times} + \delta_{+}) + \gamma_2 R(T) + O(\delta^2, \dots)
\end{split}\\
\shortintertext{where}
R(L_j) & =  \text{the condition number of the sum $L_j$}\\
& = \frac {\sum_{i=1}^{j} |\log((n-i+1)/i)|}{|L_j|} \\
R(T) & =  \text{the condition number of the sum $T = L_j + T_1 + T_2$} \\
& = \frac{|L_j| + |T_1| + |T_2|}{|T|} % \\
\end{align*}
\end{theorem}

Again we analyse the individual components separately.

\begin{lemma}[Error in Computation of Log of Binomial Coefficient]
Set
\[L_j \leftarrow  \code{log(n/1.0) + log((n-1)/2.0) + \dots{} + log((n-j+1)/j)} \]
Then $L_j$ approximates the log of the binomial coefficient $\binom{n}{j}$ with error $\delta_{L}$,
\begin{align*}
L_j & = \log\left(\binom{n}{j}\right)  \cdot (1+\delta_{L})\\
\shortintertext{with}
|\delta_{L}| & \leq  \delta_{log} + \gamma_{j-1}R(L_j) + \frac{j \delta_{\times}}{L_j} + O(\delta^2, \dots)
 \end{align*}
\end{lemma}

\begin{proof}
Each division contributes  $\delta_{\times}$, the \code{log} contributes $\delta_{log}$, and the summation
contributes $\approx (j-1)\delta_{+}$ \cite{Higham1993}.  If $j<n/2$, then $R(L_j)=1$, as all the $\log$ terms are non-negative.
\end{proof}

\begin{lemma}[Error in Computation of Log of Powers]
Let
\begin{align*}
\widehat{P}& \leftarrow   \code{d * log((a+b)/n)} \\
\shortintertext{Then}
\widehat{P}  & = \log{\left(\left(\frac{a+b}{n}\right)^{d} \right)}(1+\delta_{\widehat{P} })\\
\shortintertext{with }
|\delta_{\widehat{P} }| &\leq   \delta_{log} + \delta_{\times}+\frac{d(\delta_{\times}+\delta_{+})}{\left|\log{\left(\left(\frac{a+b}{n}\right)^d\right)}\right|} + O(\delta^2, \dots)
\end{align*}
\end{lemma}

\begin{proof}[Proof of \theoremref{AjnLogError}]
Combine the two previous lemmas and add in the contribution arising from the sum of the 3 log terms which are not all the same sign.
\end{proof}

\theoremref{AjnLogError} makes clear that the computation using \code{log} potentially has a larger relative error than using \code{pow}, due to the presence of $|T|$.
For large $n$, any (even all) of $|L_j|, |T_1|, |T_2|$ may be larger than $1024 \cdot \log(2) \approx 708$, so that there are 10 fewer bits to represent the fractional part of the output of \code{log}. For 64 bit doubles, the precision is reduced from \num{1.1e-16} to \num{2.3e-13}.
The condition number $R(T)$ can be quite large. For $x=1/\sqrt{n}$ and $j=n/2$, $A_{j}(n, x)\approx e^{-2}n^{-1/2}$, hence $T \approx -2 -0.5\log{n}$. The log binomial coefficient $L_{n/2} \approx n\log(2)-0.5\log(n)$, so $R(T) \approx (2n\log{n}-0.5\log(n) +2)/(2+0.5\log(n)) \approx \frac{4n}{\log_2{n}}$.

 %%%%%%%%%%%%%%%%%%%%%%%%%%%%%%%%%%%%
 %Grouping
 %%%%%%%%%%%%%%%%%%%%%%%%%%%%%%%%%%%%

 As an alternative to using $\log/\exp$, the computation may be done with appropriate grouping of the multiplicands of $A_j(n, x)$ to avoid overflows.
 Fix a maximum exponent $d\in\mathbb{Z^+}$. Let $j-1 = \sum_{j=0}a_i d^i$ and $n-j = \sum_{j=0}b_i d^i$ for $0\leq a_i, b_i \leq d-1$.
 Set
 \begin{equation}
 \begin{split}
% A_j(n, x)
A_{j, grp} & \leftarrow \underbrace{(n/1)* ((n-1)/2)*\dots} * \underbrace{((n-l+1)/l)* \dots} * \underbrace{\dots*((n-j+1)/j)} \\
 &*  \quad  ((j+k+\alpha)/n)^{a_0}
     * \left(((j+k+\alpha)/n)^{d}\right)^{a_1}  \\
& * \quad   \left(\left(((j+k+\alpha)/n)^{d}\right)^{d}\right)^{a_2} * \ldots \\
& * \quad  ((n-j-k-\alpha)/n)^{b_0}
    * \left(((n-j-k-\alpha)/n)^{d}\right)^{b_1}  \\
 & * \quad    \left(\left(((n-j-k-\alpha)/n)^{d}\right)^{d}\right)^{b_2} * \ldots
\end{split}
\label{eq:AjnGrouped}
\end{equation}
where the grouping of the terms for the binomial coefficient is such that the product of the terms in each group
lies between $2^{\pm1022}$,
and $d$ is such that the power terms have absolute value between $2^{\pm1022}$.
(The exponent 1022 is for 64 bit doubles. When using C's long doubles, either 80 or 128 bits, a more relevant maximum exponent is likely 16380.)
The terms need to be multiplied so as to keep the intermediate products representable.
One way to do that is to renormalize after each exponentiation and/or multiplication,
pulling out the powers of 2 and keeping track of those separately.

\begin{theorem}[Error in Grouped Computation of $A_j(n, x)$]

\label{theorem:AjnGroupedError}
Grouping terms as in \formularef{AjnGrouped},
\begin{align*}
A_{j, grp} & = A_j(n, x) \cdot (1 + \delta_{A_{j, grp}})\\
\shortintertext{with}
\begin{split}
 |\delta_{A_{j, grp}}| & \leq  (n+2j+2\log_d(n))\delta_{\times} + \left(2\frac{n-1}{d} + 2\log_d(n)\right)\delta_{pow} \\
 & \quad + (n-1)\delta_{a}  + O(\delta^2, \dots)
\end{split}
\end{align*}
\end{theorem}

\begin{proof}
The proof is similar to \theoremref{AjnError}, using \lemmaref{powererrors}\ref{singlepow} and noting that $d+1\leq 2d$ and $(j-1)(n-j)<n^2/4$.
\end{proof}

To avoid overflow of the binomial coefficient, it needs to be either rescaled or grouped into about $\frac{n}{1021}$ chunks,
as $\log\left(\binom{n}{n/2}\right) \approx n\log(2)-0.5\log(n)$.
To avoid overflow/underflow in the exponentiations, pull out the powers of 2 and take $d=512$.
In practice many of the exponentiations do not need special treatment.
For small $j$ and large $n$, $(1-x-\frac{j}{n})^{n-j} \approx e^{-(n-j)(x+\frac{j}{n})}=e^{-nx}e^{j(x-1+\frac{j}{n})}$ and this is easily within the desired bounds.
All that is required for acceptable $y^m$ is that $m\log_2(y)>-1022$.
Keeping to powers of 2, $d=512$ is the minimum value that may occur, and it may be much higher.
The value for $d$ used for one of the exponentiations can be independently chosen from the $d$ used for the other.
This implies that the $\log_d(n)$ terms in \theoremref{AjnGroupedError} bound are of much smaller importance than the other terms.

The computation using \code{log/exp} could also be grouped and the summation reordered to protect accuracy.
This only provides a benefit if the log of the binomial coefficient is also grouped, into about the same number of groups as the original grouped approach.

An alternative to grouping or using \code{log, exp} is to compute $A_j(n, x)$ using $A_{j-1}(n, x)$.

\begin{theorem}[Error in Iterative Computation of $A_j(n, x)$ using \eqnref{Ajniterative}. ]

\label{theorem:AjnIterationError}
Let $A_{j-1}$ be an estimate of $A_{j-1}(n, x)$ with relative error $\delta_{A_{j-1}}$.  Set
\begin{align*}
W & \leftarrow  \code{(n-j+1)*(k+$\alpha$+j)/(j*(n-j-k-$\alpha$))}  \\
U & \leftarrow  \code{pow(1+1/(j-1+k+$\alpha$), j-2)} \\
V & \leftarrow   \code{pow(1- 1/(n-j+1-k-$\alpha$), n-j-1)} \\
A_{j,iter} & \leftarrow  \code{W * U * V * $A_{j-1}$}\\
\shortintertext{Then}
A_{j,iter} & = A_j(n, x)  (1 + \delta_{A_{j,iter}}) \\
\shortintertext{with}
|\delta_{A_{j,iter}}| & \leq  (2n-4)\delta_{+}  + (n+2)\delta_{\times} + 2\delta_{pow} + |\delta_{A_{j-1}}| + O(\delta^2, \dots)
\end{align*}
\end{theorem}

\begin{proof}
The proof is similar to \theoremref{AjnError}.
\end{proof}

The relative error in \theoremref{AjnError} (and \theoremref{AjnGroupedError}) is bounded by $O((n+j)*(\delta_{\times}+\delta_{+}))$ which is undesirable.
This is the maximum and some of the errors might be expected, ``on average", to cancel each other out.  If the addition and multiplication errors in the binomial coefficient are  uniform (say $\text{Unif}[-0.5\delta_{\times}, 0.5\delta_{\times}]$) and independent, then their overall contribution to the error will have approximate mean 0 and variance $\frac{j}{6}\delta_{\times}^2$.  The contribution from the two powers will also have mean 0 when measured over all $x$, but a much larger variance approximately $\frac{(j-1)^2+(n-j)^2}{6}\delta_{\times}^2$.
It is clear that the bound on the final error is not so dependent on the accuracy of the standard functions \code{pow, log, exp}, but on the errors in the inputs to those function calls.
\code{pow(x$_f$, m)} itself can be accurate to within 0.5ULP, but it is the the relative error $\frac{x-x_f}{x}$ that gets multiplied by $m$ which increases the error in the output.
In order to reduce the final error, it is $\delta_{+}$ and $\delta_{\times}$ that need to be reduced.

\begin{table}[!htb]
\centering
\label{tab:errorsbydatatype}
\newlength{\signifl} \settowidth{\signifl}{Mantissa}
\newlength{\exponl} \settowidth{\exponl}{Exponent}

\begin{tabular}{lrrlllll}
\toprule
{} &  \begin{minipage}{\the\signifl}{Mantissa\\Bits}\end{minipage} &  \begin{minipage}{\the\exponl}{Exponent\\Max}\end{minipage} &  $\epsilon$ &   $\delta_{+}$ & $\delta_{\times}$ & $\delta_{pow}$ & $\delta_{log}$ \\
Name          &                &        &             &                &                   &                &                \\
\midrule
binary32      &             24 &    127 &   $2^{-24}$ &     $\epsilon$ &        $\epsilon$ &    $2\epsilon$ &    $2\epsilon$ \\
binary64      &             53 &   1023 &   $2^{-53}$ &     $\epsilon$ &        $\epsilon$ &    $2\epsilon$ &    $2\epsilon$ \\
binary128     &            113 &  16383 &  $2^{-113}$ &     $\epsilon$ &        $\epsilon$ &    $2\epsilon$ &    $4\epsilon$ \\
extended80    &             64 &  16383 &   $2^{-64}$ &     $\epsilon$ &        $\epsilon$ &    $2\epsilon$ &    $4\epsilon$ \\
double-double &            106 &   1023 &   $2^{-53}$ &  $4\epsilon^2$ &    $16\epsilon^2$ &    $$ &    $$ \\
\bottomrule
\end{tabular}
\caption{Error Bounds for some common types.}
\end{table}

\tblref{errorsbydatatype} shows values of the various $\delta$s for 3 IEEE types, the extended80 type and the software supported double-double type.
(The double-double type \cite{DoubleDoubleBailey} is a pair of binary64 values, $[a_{hi}, a_{lo}]$, with $|a_{lo}|<=2^{-53}|a_{hi}|$.)
$\delta_{log}$, $\delta_{pow}$ are dependent on the library and processor being used.  Here we use the values for the GNU C library on an i686 processor \cite{glibc2017}.
For the IEEE types, $\delta_{+}$ is the same as $\delta_{\times}$ and is $2^{-\#\text{bits in mantissa}}$, and it is usually the case that $\delta_{log}$ = $2\delta_{\times}$.
For the double-double type, $\delta_{+} = 4*2^{-\#\text{bits in mantissa}}$, and  $\delta_{\times} \approx 16*2^{-\#\text{bits in mantissa}}$.
Let \code{logD} (\code{expD} respectively) be functions that compute a double-double $\log$ ($\exp$) for double-double inputs.
The $\delta_{pow}$ listed for the double-double type, namely $2^{-52}$ comes from using the C library \code{pow} function.
Replacing \code{pow} with a fully compensated algorithm such as in \cite{graillat:ieee2009} leads to $\delta_{pow} \approx 16(m-1)\epsilon^2$.
Replacing \code{log(x)} with \code{logD(x)} leads to $\delta_{log} \approx \frac{\epsilon}{\log|x|}$.

\begin{theorem}[Error in Double-Double Computation of $A_j(n, x)$] Use double-double arithmetic to compute $A_j(n, x)$.
\label{theorem:AjnDDError}
\begin{enumerate}[label={\emph{\alph*})}]
\item Grouping terms as in \formularef{AjnGrouped}, let
    \begin{align*}
    [A_{hi},\, A_{lo}] &= A_{j, grp} \\
    \shortintertext{then}
    A_{hi} + A_{lo}  & = A_j(n, x) \cdot (1 + \delta_{A_{j, grp, dd}})\\
    \shortintertext{with}
    |\delta_{A_{j, grp, dd}}| & \leq  (36n+32j-52)\epsilon^2  + O(\epsilon^3).
    \end{align*}
In particular, $A_{hi}$ is a faithful rounding (to binary64) of $A_j(n, x)$  whenever $n<\epsilon^{-1}/136 \approx \num{6.6e+13}$.

\item Iterating as in \theoremref{AjnIterationError}  let
    \begin{align*}
    [A_{hi},\, A_{lo}]  & \leftarrow  \code{W * U * V * $A_{j-1}$}\\
    \shortintertext{Then}
    A_{hi} + A_{lo}  & = A_j(n, x)  (1 + \delta_{A_{j,iter,dd}}) \\
    \shortintertext{with}
    |\delta_{A_{j,iter, dd}}| & \leq  (40n-32)*\epsilon^2 + |\delta_{A_{j-1}}|  + O(\epsilon^3).
    \end{align*}

\item Compute \code{T} as in \theoremref{AjnLogError}, using double-double arithmetic and \code{logD}. Let
    \begin{align*}
    [A_{hi},\, A_{lo}]  & \leftarrow  \code{expD(T)} \\
    \shortintertext{If $A_{hi}$ has not underflowed, overflowed or become denormalized, then}
    A_{hi} + A_{lo}  & = A_j(n, x)  (1 + \delta_{A_{j,log,dd}})
    \\
    \shortintertext{with}
        \begin{split}
        |\delta_{A_{j,log,dd}}| & \leq  \delta_{exp} +3 |T| \delta_{log} + \\
         & \quad (8 |T| R(L_j) + 32 |T| + 24 R(T) + 20n-20)\epsilon^2  + O(\epsilon^3).
        \end{split}
    \end{align*}
\end{enumerate}
\end{theorem}

\begin{proof}
\begin{enumerate}[label={\emph{\alph*})}]
\item
By implementing a double-double \code{powD} function which renormalizes as it proceeds, it is possible to assume that $d$ is larger than $n$, in which case the estimate of
\lemmaref{powererrors}\ref{singlepow} is applicable.  Use $\delta_{pow} \approx 16(j-1)\epsilon^2$ (or $16(n-j)\epsilon^2)$).
\end{enumerate}
\end{proof}

The definitions of \code{logD} and \code{expD} for double-doubles are just the \ordinalnum{1} order corrections arising from the Taylor series expansions of (the inverse of) these functions.
As $\code{log}$ is a many-to-one function for much of its domain $\mathbb{F}$, the ${}_{lo}$ part of the double-double should provide some more resolution.
In practice it is not quite so simple as the library $\code{exp}$ is used to create a correction for $\code{log}$, (and $\code{log}$ to create a correction for $\code{exp}$) and neither of these functions is guaranteed to round correctly. Nor are they guaranteed to be consistent with each other
(Is $\log(\exp(\log(x)) \overset{?}{=} \log(x)$).
Even though part of the error of $A_j(n, x)$ is now $O(\epsilon^2)$, the impact of $ \delta_{exp} +3 |T| \delta_{log}$ leads to errors much bigger than observed with either powers or iteration.

Computing using double-double arithmetic is quite a bit more expensive than using the hardware supported double type.  Just adding two doubles to create a double-double involves 6 flops and multiplying two doubles to create a double-double involves 17 flops and multiplying two double-doubles involves 37 flops (though with the advent of fused-multiply-accumulate (FMA) instructions these counts have been lowered.)
In particular exponentiation is noticeably much slower, as it requires $~2\log_2(n)$ multiplications of double-doubles.

A less expensive alternative is to replace this one operation with a simpler approximation  \code{powDSimple},
which applies a Taylor Series adjustment on top of the ``C'' library function  \code{pow}.
% \code{add2(A,B)} returns

\setcounter{algorithm}{1}
\begin{algorithm}[H]
\caption{powDSimple: Calculate $(A_{hi}+A_{lo})^m$}
function $[X_{hi}, X_{lo}] = \code{powDSimple}(A, m)$
\label{alg:powDSimple}
\begin{algorithmic}
\REQUIRE $A$ a double-double \AND $m \in \mathbb{Z}$
\STATE Y $\leftarrow$ \code{pow($A_{hi}$, m)}
\STATE Z $\leftarrow$ \code{$A_{lo}/A_{hi}$}
\IF {m $> \num{1e8}$}
\STATE W $\leftarrow$ \code{expm1(m * log1p(Z))}
\ELSE
\STATE W $\leftarrow$ \code{m*Z * (1 + (m-1)*Z/2)}
\ENDIF
\STATE $X_{hi}, X_{lo} \leftarrow$ \code{add2(Y, Y*W) }
\RETURN ($X_{hi}, X_{lo}$)
\end{algorithmic}
\end{algorithm}

\begin{theorem}[Error in Computation of $A_j(n, x)$ with simplified \code{pow}]
\label{theorem:AjnPowError}
Group terms as in \formularef{AjnGrouped}, use double-double arithmetic, with \code{powDSimple} as above and $d=512$
\begin{align*}
[A_{hi},\, A_{lo}] &= A_{j, grp} \\
\shortintertext{then}
A_{hi} + A_{lo}  & = A_j(n, x) \cdot (1 + \delta_{A_{j, grp, dd}})\\
\shortintertext{with}
|\delta_{A_{j, grp, dd}}| & \leq  \left(\frac{4(n-1)}{d} + 2\log_d{n}\right)\epsilon + O(\epsilon^2).
\end{align*}
\end{theorem}

This is a substantial improvement over \theoremref{AjnGroupedError} at the cost of some double-double arithmetic.
Because the $d$ used may be substantially larger than 512, the relative error is often lower than that promised.

%%%%%%%%%%%%%%%%%%%%%%%%
%%%%%%%%%%%%%%%%%%%%%%%%
%%%%%%%%%%%%%%%%%%%%%%%%

%%%%%%%%%%%%%%%%%%%%%%%%
%%%%%%%%%%%%%%%%%%%%%%%%
%%%%%%%%%%%%%%%%%%%%%%%%

\subsection[Computation of a single A_j(n, x) term]{Computation of a single $A_j(n, x)$ term}
\label{subsec:smirnov_prop_Ajnx}

Based on the error analysis, calculating using powers is the most accurate.
As the exponents will be large, we will need to assume that the function that
computes $\left(\frac{a+b}{c+d}\right)^m$  (which we'll denote $\powfour(a,b,c,d,m)$) returns not a single number
but a pair $(y, E_y)$, such that $E_y$ is an integer, $0.5 \leq |y| \leq 1$ and $y 2^{E_y}$ is the actual computed value.
I.e. The significand and exponent are returned separately, and the exponent has a much bigger range than the usual range offered by 64 bit floats.
The C function \code{frexp} will do this separation for any 64 bit float, while \code{ldexp(x, expt)} will do the opposite.

\begin{algorithm}[H]
\caption{Compute $A_j(n,x), D_j(n,x)$ using double64  }
function $[A_j, D_j] = \code{smirnovAj}(n, k, x, \alpha, j, C, E_C)$
\begin{algorithmic}
\REQUIRE $n \in \mathbb{Z^+}$ \AND $0 < x < 1$ \AND  $j \in \mathbb{Z}$ \AND $0 \leq j \leq n$
\REQUIRE $C$ is a double-double and $(C_{hi} + C_{lo})\cdot 2^{E_C} \approx$  the binomial coefficient $\binom{n}{j}$
\STATE    $S, E_S  \leftarrow   \powfour(j+k, \alpha, n, 0, j-1) $
\STATE    $T, E_T   \leftarrow  \powfour(n-j-k,  -\alpha, n, 0, n-j) $
\STATE     $A_j   \leftarrow \code{ldexp($C * T * S$, $E_C + E_S + E_T$)} $
\STATE $W \leftarrow {\left(n-j - k\right)}* {\left(j + k\right)}* n -k*{\left(k *n + n - j - k\right)} $
\STATE  $Z 	\leftarrow n*(2*(j+2*k)-n +1)-(j+2*k)$
\STATE    $M  \leftarrow ( W + Z *\alpha + \left(2 * n - 1\right) *\alpha * \alpha )/
    {\left(j + k - n + \alpha\right)} / {\left(j + k + \alpha\right)} / (k+\alpha)$
\STATE    $D_j  \leftarrow  M \ * A_j$
\RETURN $A_j, D_j$
\end{algorithmic}
\end{algorithm}

%%%%%%%%%%%%%%%%%%%%%%%%%%%%%%%%%%%%%%%
%%%%%%%%%%%%%%%%%%%%%%%%%%%%%%%%%%%%%%%

\subsubsection{Remarks}

The simplicity of the algorithm hides a lot of complexity.
A few approaches were tried before settling on the renormalizing exponentiation approach.

\begin{itemize}

\item
Using $\log/\exp$ always gave an answer, but the answer quickly lost accuracy for two reasons.
The input to $\log$ was only an approximation, whose relative error was then multiplied when it passed through $\exp$.
The output of 3 $\log$ operations were added and suffered subtractive cancellation.

\item
Using an iterative approach to generate $A_j(n, x)$ from $A_{j-1}(n, x)$ has a relative error inherited from $A_{j-1}(n, x)$ and an additional one from the iteration calculation.
The iterative factor involves numbers very close to 1, such as $(1+\frac{1}{j-1+k+\alpha})$, and they are somewhat robust to approximations in the denominator.
However the
exponentiation of these numbers is not, and the calculations need to be performed in extended precision
 or the relative error quickly grows.
 The approach also needs a non-zero $A_{j-1}(n, x)$ to prime the iteration, either from an exponentiation or a $\log/\exp$.  The non-zero $A_{j-1}(n, x)$ also needs to be not denormalized, as otherwise that initial reduced precision will propagate throughout the rest of the $A_{j}(n, x)$.

\item
The expression  $\left(\frac{n-j-k -\alpha}{n}\right)^{n-j}$ is prone to underflow and/or denormalization.
The appearance of 0 as an intermediate results was a major source of large absolute errors in $S_n(x)$:
  important terms in the summation were just not there.
The appearance of denormalized numbers as intermediate results was a major source of large relative errors in $S_n(x)$.
When one of the two power terms was denormalized, that loss of precision carried over to the product with the binomial coefficient.
Splitting the exponentiation into two smaller pieces, such as $\left(\frac{n-j-k -\alpha}{n}\right)^{\frac{n-j}{2}}$, extended the range of computable $S_n$ to $n\approx 2000$.
Splitting into 3 pieces gave no benefit unless the binomial coefficient was also broken into a similar number of pieces.  Hence the need for a self-normalizing pow function and a normalized binomial coefficient.

 \item If $\left(\frac{n-j-k -\alpha}{n}\right)^{n-j}$ is computed as $\code{pow}(\fl((n-j-k -\alpha)/n), n-j)$ accuracy is lost even for moderate $n$, as $\fl((n-j-k -\alpha)/n)$ is only an approximation to $\frac{n-j-k-\alpha}{n}$, and the relative error of the result is magnified by $n-j$.
 Doing this computation in 80 bit or 128 bit floating point types is greatly beneficial.

 \item The binomial coefficient $\binom{n}{j}$ was also a cause for lost accuracy, as might be expected given that it is the result of $j$ multiplications and $j$ divisions. One might hope that the errors are randomly distributed, both positive and negative, so that the overall effect was negligible, but that was not the case.
 Use double-double arithmetic to preserve.

\item
 \figref{smirnov_regions} shows the algorithm associated with each region of  $\mathbb{N} \times [0,1]$.
The top/bottom/left boundaries are the special cases $x=0,1$ and $n=1$.
The $\SF$ in the far right region is too small for 64 bit doubles, so 0.0.
Smirnov/Dwass is used for small $nx$,  Smirnov for the rest.

The location of underflows of $A_j(n, x)$ are shown for $j=0, 1, 2, 10$.
If $A_0(n, x)$ underflows, then all $(1-x-j/n)^{n-j}$ will underflow.
This red region is quite pervasive for larger values of $n$, and eventually almost all $x$ belong to this region.
For each $j\ge 1$, there is an (initially thin) region just below $x=1-j/n$ which results in $(1-x-j/n)^{n-j}$ underflowing.
Not all of $(n, x)$ pairs in the region are important, as even if $A_j(n, x)$ was calculated correctly
it would not contribute significantly to the final answer.
But some are quite relevant, even for values of $n$ much less than 1030,
which is where the binomial coefficient overflows.

This approach avoids any abrupt changes as the underlying algorithm doesn't change as $n$ becomes larger.

\begin{figure}[!htb]
\centering
\includegraphics[scale=0.33]{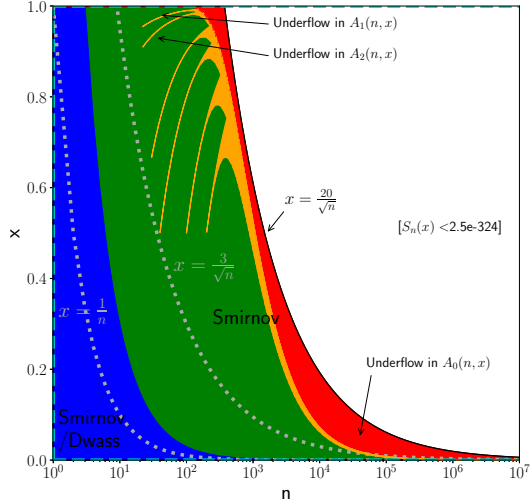}
\caption{Method used to compute $S_n(x)$ in each region of $\mathbb{N} \times [0,1]$, along with a few regions of intermediate computation underflow.}
\label{fig:smirnov_regions}
\end{figure}

%\item Checks on the inputs are not presented here.  It is expected that $n$ is a positive integer, and $x$ is a real number $0 \leq x \leq 1$.
 \end{itemize}

\subsubsection{Use of double double precision}

If double-double operations are available, the use of $\alpha$ and $k$ can be dropped and an algorithm used which more closely follows the mathematical description.  It requires a function \code{powScaledD(A, m)} to compute the (integer) power of a double-double together with a scaling, which can be obtained by the obvious modifications of \code{CompLogPower}\cite{DBLP:journals/corr/abs-0705-4369} or \code{LogPower)}\cite{graillat:ieee2009}, though cheaper alternatives work almost as well.

We'll represent a double-double $A$ as a pair $[A_{hi}, A_{lo}]$, with $| A_{lo} | \leq \epsilon | A_{hi} | $.
The routine  \code{div22(a, b)} takes two doubles and returns a double-double approximation to the real number $a/b$.
The routines \code{addDD}, \code{mulDD} \code{divDD} take as input a pair of double-doubles and return a double-double, being (an approximation to) the sum/product/quotient respectively. The routines \code{addD2}, \code{divD2} replace the  \ordinalnum{2} input with a double.

\begin{algorithm}[!htb]
% Allow if to move to prevent a big gap 
%\begin{algorithm}[H]
\caption{Compute $A_j(n,x), D_j(n,x)$
using double-doubles }
function $[A, D)] = \code{smirnovAjD}(n, x, j, C, E_C)$
\begin{algorithmic}
\label{alg:smirnov_alg_dd}
\REQUIRE $n \in \mathbb{Z^+}$ \AND $0 < x < 1$ \AND  $j \in \mathbb{Z}$ \AND $0 \leq j \leq n$
\REQUIRE $C$ is a double-double and $(C_{hi} + C_{lo})\cdot 2^{E_C} \approx$  the binomial coefficient $\binom{n}{j}$
\STATE $ P \leftarrow \code{addD2(div22(j, n), x)}$
\STATE $ Q \leftarrow \code{addD2(div22(n-j, n), -x)}$
\STATE $ S, E_S  \leftarrow   \code{powScaledD(P, j-1)} $
\STATE $ T,  E_T  \leftarrow   \code{powScaledD(Q, n-j)} $
\STATE $ A \leftarrow   \code{mulDD(mulDD(C, T), S)} $
\STATE $ E \leftarrow E_C + E_S + E_T$
\STATE $ A   \leftarrow [\code{ldexp}(A_{hi}, E), \; \code{ldexp}(A_{lo}, E)] $
\STATE $ M \leftarrow \code{addDD(div22(1, x), div2D(j-1, P))}$
\STATE $ M \leftarrow \code{addDD(M, div2D(j-n, Q)) } $
\STATE $ D \leftarrow \code{mulDD(M, A)} $
\RETURN $(A, D)$
\end{algorithmic}
\end{algorithm}

Even though the double-double $P$ in \algref{smirnov_alg_dd} may not be exactly equal to $x+\frac{j}{n}$, its relative error is at most $4\epsilon^2$, so that the propagated error after exponentiation is small enough.
Exponentiating a double-double can be computationally intensive for large $m$.
Using a simpler algorithm such as \code{powDSimple}  does not have the same accuracy as a full-blown double-double exponentiation routine,
but it requires much less computation,
and it has removed the bias arising from a non-zero $P_{lo}$.

%%%%%%%%%%%%%%%%%%%%%%%%%%%%%%%%%%%%%%%%%%%%%%%%%%%
%%%%%%%%%%%%%%%%%%%%%%%%%%%%%%%%%%%%%%%%%%%%%%%%%%%
%%%%%%%%%%%%%%%%%%%%%%%%%%%%%%%%%%%%%%%%%%%%%%%%%%%

    %!TEX root = ./ms.tex
\section{Computation of the Inverse  Survival Function}

\label{sec:smirnovi}

There is no formula to invert the Survival Function \eqnref{sm_1}, so computing the inverse of $S_n(x)$ has to be done numerically, either by interpolation from tables, or some form of root-solving.

%%%%%%%%%%%%%%%%%%%%%%%%%%%%%%%%%%%%%%%%%%%%%%%%%%%
%%%%%%%%%%%%%%%%%%%%%%%%%%%%%%%%%%%%%%%%%%%%%%%%%%%
%%%%%%%%%%%%%%%%%%%%%%%%%%%%%%%%%%%%%%%%%%%%%%%%%%%

%%%%%%%%%%%%%%%%%%%%%%%%%%%%%%%%%%%%%%%%%%%%%%%%%%%
%%%%%%%%%%%%%%%%%%%%%%%%%%%%%%%%%%%%%%%%%%%%%%%%%%%
%%%%%%%%%%%%%%%%%%%%%%%%%%%%%%%%%%%%%%%%%%%%%%%%%%%

\subsection{Initial considerations}

The function $S_n(x)$ is monotonic, the derivative is never 0 inside the open interval $(0, 1)$, so the expectation is that it should be straight forward to invert
\begin{align}
\label{eq:Sn_equals_PSF}
S_n(x) & = p_{SF}
\end{align}
However there are some complications.
$S_n$ is a spline of polynomials and these are all infinitely differentiable, hence $S_n$ is $C^\infty$ in between the knots.
But $S_n$ is only $C^{j-1}$ at $x=\frac{j}{n}$.
In particular, $S_n$ is continuous but not differentiable at $x=\frac{1}{n}$.
Additionally, the contributions to the Taylor Series for $S_n$ from the \ordinalnum{2} and higher order terms are non-negligible.
For $x \geq \frac{n-1}{n}$
\begin{align}
S_n(x) & = (1-x)^n \\
\frac{d^j S_n}{dx^j}(x) & = (-1)^j \frac{n!}{(n-j)!}(1-x)^{n-j}
\\ \shortintertext{so that}
\frac{S_n''(x)}{S_n'(x) } & = -\frac{(n-1)}{(1-x)}
\end{align}
and the factor $S_n''/S_n'$ is unbounded as $x\rightarrow 1$.
Similarly $S_n''/S_n' \rightarrow (n-1)$ as $x\rightarrow 0$.
Hence the endpoints of the interval are potential causes of trouble.
The first of these, $x$ close to 1 (so that $p_{SF}$ close to 0), can be dealt with immediately: 
if $p_{SF} \leq n^{-n}$, then $x= 1-\sqrt[n]{p_{SF}}$.

It may be tempting to use the asymptotic $S_n(x) \approx T_n(x) \triangleq T(\sqrt{n}x)$  where $T(x) \triangleq e^{-2x^2}$,
as the computational cycles required are much less than for a computation of $S_n(x)$.
Certainly as $n \rightarrow \infty$, $S_n(x/\sqrt{n})/T(x) \rightarrow 1$.  But there are some caveats.
\begin{itemize}
    \item For small $n$, the $S_n$ and the asymptote are really not so close.
    \item The domain of $T_n$ is $\mathbb{R}$, whereas the domain of $S_n(x/\sqrt{n})$ is only $[0, \sqrt{n}]$,
    so the range of $T_n^{-1}$ is a proper superset of the desired interval.
    \item For all $n$, $\lim_{x\to0}S_n'(x) = -1$, whereas the limit of the derivative of the
    asymptotic,  $\lim_{x\to0}\frac{dT_n}{dx} = \lim_{x\to0} -4nxe^{-2nx^2} = 0$.
        \item The derivative of $S_n$ is not continuous, having a jump at $x=\frac{1}{n}$ of size $1$.
        (See \figref{Sn_prime}. )
\end{itemize}

\begin{figure}[!htb]
\includegraphics[scale=0.5]{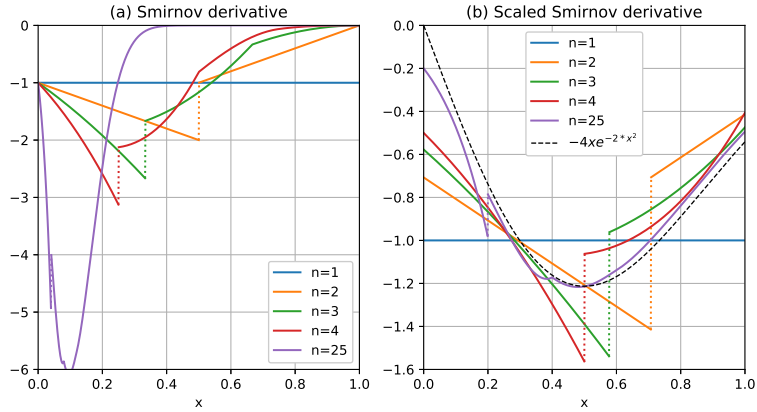}
\caption{(a) Derivative of $S_n(x)$ for $n=1, 2, 3, 4, 25$. (b) Scaled derivatives of smirnov probabilities:  
$\frac{1}{\sqrt{n}}S_n'(\frac{x}{\sqrt{n}})$, along with the derivative of the asymptote, $-4xe^{-2x^2}$.}
\label{fig:Sn_prime}
\end{figure}

%%%%%%%%%%%%%%%%%%%%%%%%%%%%%%%%%%%%%%%%%%%%%%%%%%%
%%%%%%%%%%%%%%%%%%%%%%%%%%%%%%%%%%%%%%%%%%%%%%%%%%%
%%%%%%%%%%%%%%%%%%%%%%%%%%%%%%%%%%%%%%%%%%%%%%%%%%%

%%%%%%%%%%%%%%%%%%%%%%%%%%%%%%%%%%%%%%%%%%%%%%%%%%%
%%%%%%%%%%%%%%%%%%%%%%%%%%%%%%%%%%%%%%%%%%%%%%%%%%%
%%%%%%%%%%%%%%%%%%%%%%%%%%%%%%%%%%%%%%%%%%%%%%%%%%%
\subsection{Estimating the root with Newton-Raphson}

The Newton-Raphson (N-R) root-finding algorithm can be used to find the root of $S_n(x)=p$, with multiple iterations of
\begin{align}\label{eq:NR}
 x & \leftarrow x - \frac{f(x)}{f'(x)}
\end{align}
with $f(x) = \smirnov(n, x)-p$.

The N-R algorithm requires a starting estimate $x_0$.
Inverting the \eqnref{SmirnovAsymptote} approximation $S_n(x) \approx e^{-2nx^2}$ leads to  $x_0 = T_n^{-1}(p) = \sqrt{\frac{-\log p}{2n}}$.
This estimate is always to the right of the root.
If $p <  e^{-2n}$, the initial estimate lies outside the interval $[0, 1]$.

Even with a valid initial estimate it is still possible to be driven outside the interval, or to enter a (almost-)cycle without converging.
 N-R generates successively better estimates for the root whenever the function is well approximated by the tangent line in a neighbourhood containing the root and the current estimate.
For successful N-R root-finding, the errors satisfy $e_{n+1}\approx \frac{S_n''}{2S_n'} e_{n}^2$, but we've already noted that $S_n''/2S_n'$ can be very large, or even unbounded.

N-R has a tendency to overshoot the root if $f(x) * f''(x) < 0$ and otherwise stay on the same side of the root.
(The $2^{nd}$ derivative of $S_{n}(x) \approx (16n^2x^2-4n)e^{-2nx^2}$, which is negative for $x <\frac{1}{2\sqrt n}$.)
Ideally our initial estimate of the root would be less than the root if $x \gtrapprox \frac{1}{2\sqrt n}$, and greater than the root otherwise.

If $p$ is close to 0, and $x$ is to the right of the root, the tangent line may far overshoot and intersect the $x$-axis to the left of $x=0$.  Even if $x$ remains in the domain, the new estimate may be much further away than the current estimate.
Since $x_0$ is always greater than the root, some amount of overshoot will always occur for $p \lessapprox e^{-0.5} \approx 0.6$.
To illustrate the phenomenon we show here (\tblref{NRTable}) the iterates for $p=0.000001055$ and $n=10$ with initial estimate = $\sqrt{-\log(p)/2n)} = 0.829517$.
(The desired root is $x=0.753671966\ldots$.)
\begin{table}[!htb]
\centering
\begin{tabular}{ r| r | r | r }
Iteration & x & smirnov(10, x)  & smirnov(10, x) -p \\
\hline
0 & 0.829517 & 2.109618e-08 &-0.000001\\
1 & 0.009993 & 9.890717e-01 & 0.989071\\
2 & 0.840449 & 1.076972e-08 & -0.000001\\
3 & -0.690884 & NaN & NaN \\
\hline
\end{tabular}
\caption{Newton-Raphson iterates}
\label{tab:NRTable}
\end{table}

Using the derivative of the asymptote, $\frac{dT_n}{dx}(x)$, as a substitute for $S_n'(x)$ will result in slightly smaller steps being taken for $x \gtrapprox 0.5/\sqrt{n}$,
and larger ones for $0 < x \lessapprox 0.5/\sqrt{n}$, leading to a lower convergence rate.
For $x \leq \frac{1}{n}$, the result will be more catastrophic,  likely overshooting the root, and even overshooting the domain.

For pure N-R to work well, the initial estimate needs to be close enough to the root, or on the correct side of the root, so that convergence is guaranteed.
(Inverting \eqnref{MaagAsymptote} leads to $x_0 = \sqrt{\frac{-\log p}{2n}} - \frac{1}{6n}$ which is a better estimate.)
Or it needs to be augmented with a bracket.  Of course, the bracket $[0, 1]$ is initially available, even if it is not particularly constraining.

%%%%%%%%%%%%%%%%%%%%%%%%%%%%%%%%%%%%%%%%%%%%%%%%%%%
%%%%%%%%%%%%%%%%%%%%%%%%%%%%%%%%%%%%%%%%%%%%%%%%%%%
%%%%%
%%%%%%%%%%%%%%%%%%%%%%%%%%%%%%%%%%%%%%%%%%%%%%

%%%%%%%%%%%%%%%%%%%%%%%%%%%%%%%%%%%%%%%%%%%%%%%%%%%
%%%%%%%%%%%%%%%%%%%%%%%%%%%%%%%%%%%%%%%%%%%%%%%%%%%
%%%%%%%%%%%%%%%%%%%%%%%%%%%%%%%%%%%%%%%%%%%%%%%%%%%

\subsection{The SciPy implementation}

The SciPy implementation of $\smirnovi(n, p)$ uses an unbracketed Newton-Raphson algorithm, always initializing from the asymptotic behaviour \eqnref{SmirnovAsymptote} $S_n(x) \approx T_n(x) = e^{-2nx^2}$, and always using $T_n'(x)$ as a substitute for $S_n'(x)$.  It halts whenever one of: $x$ is outside the interval $[0,1]$; the number of iterations exceeds 500; the relative change in successive $x$-iterates is less than $10^{-10}$.

\begin{figure}[!htb]
\centering
\includegraphics[scale=0.5]{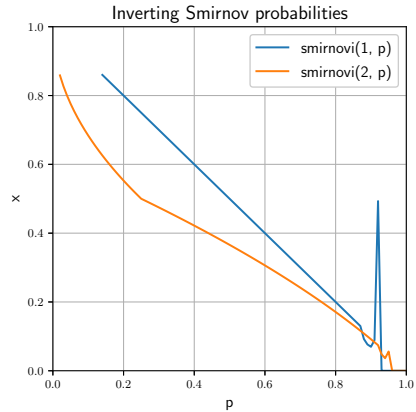}
\caption{Inverting $\smirnov$}
\label{fig:smirnovi}
\end{figure}

A plot (\figref{smirnovi}) of $\smirnovi(n, p)$ for $n=1, 2$ shows some missing values and oscillations near the endpoints of the interval.

\begin{itemize}
    \item The leftmost missing values near $p=0$ are due to the initialization procedure generating out-of-domain $x_0 > 1$.
    \item The next leftmost missing values near $p=0$ are due to to early overshoot ($x_0$ is in domain, but $x_1< 0$).
    These missing values can occur for all $n$.

    \item The rightmost values near $p=1$ are missing, due to the use of the approximate derivative: $x_0$ is in domain, but the iteration overshoots and $x_1 < 0$.

    \item The oscillations near $p=1$ are also due to the use of approximate derivative, as the maximum number of iterations may be exceeded without convergence and the current estimate returned.
    In particular, $\smirnovi(n,x)$ is not close to being monotonic,
    which can be source of unexpected behaviour when it is used to generate random variates with the $S_n(x)$ distribution.

    \item When the iterations do converge, the convergence rate is often closer to linear than quadratic, due to the use of the approximate derivative --- it just isn't close enough to the actual derivative for the faster convergence.  To achieve the same tolerance, using the approximate derivative needed about 10x as many iterations.
\end{itemize}

These phenomena appear most prominently when used with small $n$, or {\em extremely small} $p$, or $p$ very close to 1.  If the interest is only in large gaps in the EDF, they are less likely to appear.
This may be why the phenomena have not been noticed/addressed before now.

    %!TEX root = ./ms.tex

%%%%%%%%%%%%%%%%%%%%%%%%%%%%%%%%%%%%%%%%%%%%%%%%%%%
%%%%%%%%%%%%%%%%%%%%%%%%%%%%%%%%%%%%%%%%%%%%%%%%%%%
%%%%%%%%%%%%%%%%%%%%%%%%%%%%%%%%%%%%%%%%%%%%%%%%%%%

\section{Algorithm for Computing Smirnov Quantile}
\label{sec:smirnovi_prop} 
Next we propose modifications to the existing algorithm
which will find $x$ such that $\smirnov(n, x) = p$ within the specified tolerances.

\setcounter{pvmalgorithm}{\value{algorithm}}
\begin{pvmalgorithm}[smirnovi]
\label{alg:smirnovi_alg}
\label{smialgapi}
Compute the Smirnov $\ISF/\PPF$ for integer $n$ and real $p_{SF}$, $p_{CDF}$.
\end{pvmalgorithm}
\setcounter{algorithm}{\value{pvmalgorithm}}
\code{function [$X$] = smirnovi(n:int, pSF:real, pCDF: real)}
\noindent
\begin{enumerate}[start=1,label={{\bf Step \arabic*}},ref={Step \arabic*}]

\item Immediately handle $p_{SF}=0$, $p_{CDF}=0$ as special cases, returning $X \leftarrow 1$ or $0$ respectively.
\label{smialg_x0}

\item Immediately handle $n=1$  as a special case, returning $X \leftarrow p_{CDF}$.
\label{smialg_n1}

\item Immediately handle $0 < p_{SF} \leq n^{-n}$.
Return $X \leftarrow  1 - \pow(p_{SF},  1 / n)$.
\label{smialg_case0}

\item Set $P_1 \leftarrow  \frac{1}{n}(1+\frac{1}{n})^{n-1}$.
If $p_{CDF}  <= P_1$, set
    \begin{align}
    [A, B] & \leftarrow   {} [\frac{p_{CDF}}{e}, \, \min(p_{CDF}, \frac{1}{n})] \notag \\
    \gamma_0  & \leftarrow   {} \frac{p_{CDF}} {P_1}  \notag\\
    \gamma_1 & \leftarrow   {} \frac{\gamma_0(\gamma_0 + \exp(1-\gamma_0))}{\gamma_0+1} \notag\\
    X_0 & \leftarrow   {}\min(\frac{\gamma_1}{n}, B)
   \label{eq:X0smallx}
    \end{align}
\label{smialg_casei}

\item 
\label{smialg_caseii}
 If $p_{CDF}  > P_1$, set
\begin{align}
A & \leftarrow  \max\left (1 - \pow(p_{SF}, 1 / n), \; \frac{1}{n}\right) \notag \\
B_0 & \leftarrow  \sqrt{-\log(p_{SF})/(2n)}\notag \\
B_1 & \leftarrow  B_0 - \frac{1}{6n} \notag \\
B  & \leftarrow  \min(B_0, 1-\frac{1}{n})\notag \\
\label{eq:X0}
X_0 & \leftarrow  \begin{cases}
    B_1                 & \text{if $A \leq B_1 \leq B$} \\
    \frac{A+B}{2} & \text{otherwise}
    \end{cases}
\end{align}

\item
\label{smialg_funcdef}
Define the function
\begin{equation*}
    f(x)= \begin{cases}
    \code{smirnov(n, x).SF} - p_{SF} & \text{if $p_{SF}<= 0.5$}\\
    p_{CDF} - \code{smirnov(n, x).CDF} & \text{otherwise}
    \end{cases}
\end{equation*}

\item Perform iterations of bracketed N-R with function $f$, starting point $X_0$ and bracketing interval $[A, B]$, until the desired tolerance is achieved, or the maximum number of iterations is exceeded.  Set $X \leftarrow$ the final N-R iterate.  Return $X$.
\label{smialg_nr}

\end{enumerate}

%%%%%%%%%%%%%%%%%%%%%%%%%%%%%%%%%%%%%%%%%%%%%%%%%%%
%%%%%%%%%%%%%%%%%%%%%%%%%%%%%%%%%%%%%%%%%%%%%%%%%%%
%%%%%%%%%%%%%%%%%%%%%%%%%%%%%%%%%%%%%%%%%%%%%%%%%%%
\subsection{Remarks}
\label{smirnovi_remarks}
\SciPy does contain multiple root-finders but we avoid
using them here.
The code for \smirnovcode{} is written in C as part of the {\code{cephes}} library in the \code{scipy.special} sub-package.
In order to enable an implementation of this K-S algorithm to remain contained within this sub-package,
we only use a bracketing Newton-Raphson root-finding algorithm, as this can be easily implemented.
\begin{itemize}
\item

 $f$ is $C^1$ inside any of the bracketing intervals $[A, B]$, as $\frac{1}{n}$ is not an
 interior point of any such interval, hence use of a order 1 method is justified.
 Use of the actual derivative requires implementing a \code{smirnovp} function to calculate the derivative,
 which can be calculated with the incremental steps to \smirnovcode{} as \eqnref{dSn}.
\end{itemize}

\begin{itemize}
\item
Specifying the API to take both the \SF\/ and \CDF\/
enables the code to use which ever probability allows the greatest precision, which will usually be the smaller of the two.
[If needed the original API can also easily be preserved by giving the new function a new name (such as {\code{\_smirnovi}}) and
chaining {\code{smirnovi(n, p)}} to call {\code{\_smirnovi(n, p, 1-p)}}.]
\end{itemize}

After the obvious endpoints have been handled in \stepref{smialg_x0} and \stepref{smialg_n1},
any remaining root $x$ lies in one of the three intervals
 $(0 , \frac{1}{n}]$ , $(\frac{1}{n}, \frac{n-1}{n})$ and $[\frac{n-1}{n}, 1)$
  and each interval lends itself to a different approach.

\begin{itemize}
\item \stepref{smialg_case0} handles the $3^{rd}$ interval in the list.
\begin{align}
\Prob \{D_n \geq x\}  \leq \frac{1}{n^n} \iff x = 1-(\Prob \{D_n \geq x\})^{\frac{1}{n}}.
\end{align}

\end{itemize}
%\item
The other two intervals will be handled by numerical root-finding, in particular a combination of fixed-point iteration and Newton-Raphson methods.
An initial estimate is generated by prior analysis or some iteration converging to a fixed point of an auxiliary function,
which is then used as the input to the N-R.
\begin{itemize}

\item  \stepref{smialg_casei} uses
$\Prob \{D_n \leq x\} = x(1+x)^{n-1}$ for $0 < x \leq \frac{1}{n}$
 (\eqnref{Sn_upper1}).
\begin{align}
P_1 & = \Prob \{D_n \leq \frac{1}{n}\}\qquad (\approx \frac{e}{n+1})
\\  \frac{\Prob \{D_n \leq \gamma  \frac{1}{n}\}}{\Prob \{D_n \leq \frac{1}{n}\} }
   & = {} \frac{\gamma(1+\gamma/n)^{n-1}}{(1+1/n)^{n-1}}
\\ & = {} \gamma e^{\gamma - 1}  \left(1+O(\frac{1}{n})\right)\quad \text{for $0 < \gamma < 1$}
\\
\intertext{If}
\label{eq:gammaegamma}
\gamma e^{\gamma - 1} & = \frac{p_{CDF}}{P_1}
\end{align}
then $\frac{\gamma}{n}$ is an approximation to the root of \eqnref{Sn_equals_PSF}.
There is no closed-form solution for $\gamma$, but starting with $\gamma_0=\frac{p_{CDF}}{P_1}$ and applying one iteration of Newton-Raphson to \eqnref{gammaegamma}
 leads to $X_0$ as in \eqnref{X0smallx}.
 This $X_0$ is greater than the root of \eqnref{Sn_equals_PSF}, hence on the ``good side" for N-R iterations.

\item For $p_{CDF} \leq P_1$, an alternative approach is to iterate
$g_{n,p}(x) \triangleq \frac{p_{CDF}}{(1+x)^{n-1}}$
starting from $x=\frac{p_{CDF}}{n P_1}$, as $g_{n,p}$ is contractive around its fixed-point.
If the number of iterations is an odd number, then the final iterate will be greater than the root, leaving it in good stead for the later N-R stage.
This alternative approach does take $n$ into consideration but is a little more complicated
and provides little additional value, as the contraction factor for $g_{n,p}(x)$ can be very close to 1.

\item
\stepref{smialg_caseii} handles the majority of $x$ values.  It uses the asymptotic behaviour to generate a starting point and a bracket.
[As usual, whenever computing $\log(p_{CDF})$, use $\log(p_{CDF})$ if $p_{CDF} < 0.5$, and $\logonep(-p_{SF})$  if $p_{CDF} >= 0.5$.  Similarly for $\log(p_{SF})$.]
Note that $A >= \frac{1}{n}$ and $B <= \frac{n-1}{n}$, so that $X$ is kept away from the problematic endpoints.
Using $B_0$ as the starting point would lead to a significant overshoot problem.
Using $B_1$ doesn't guarantee that overshoots will not occur, as $B_1$ can be on either side of the root,
but makes drastic overshoot much less likely.

\item
The two expressions in \stepref{smialg_funcdef} would compute the same answer if using infinite precision ---
any difference between them should be approximately the order of the machine epsilon.
The two expressions obviously have the same derivative, which can be implemented with the obvious minor modifications to \smirnovcode.

\item In \stepref{smialg_nr} use the actual derivative, not an estimate based on the limiting behaviour.
 If the N-R delta is ``too big" relative to the delta of two steps ago, perform a bisection step.

Any non-zero difference between successive iterates is bounded below by the product of the minimum non-zero values of $\left| f(x)-p\right|$ and $|1/f'(x)|$.
As $x \rightarrow 0$, $|f'(x)| \rightarrow 1$, so it is important to use a formulation in which the first term is small.
\Ie to calculate the ( $p_{CDF} - \text{cumulative probability}$) rather than ($\text{survivor probability} -  p_{SF}$).

\item
In practice it was found prudent to increment $A$ slightly so as to ensure that the computed $\smirnov(n, A) - p_{SF} \geq 0$.
Replacing $A$ with $A*(1-256\epsilon)$ was found sufficient, where $\epsilon$  is the machine epsilon.
Similarly replace $B$ with $B*(1+256\epsilon)$, especially if the root-finder requires $f(A)*f(B)<0$.
As the bracket $[A, B]$ is updated quickly during the N-R steps, its main use is to constrain the initial value $X_0$ as in \eqnref{X0}, hence the adjustment has little impact on the accuracy or number of iterations.

\end{itemize}

The main substance to the algorithm is the determination of a suitable bracket and initial starting point.
Once that has been done, many suitable root-finders exist.  Some results using alternative root-finders are discussed in \sref{smirnovi_results}.

%%%%%%%%%%%%%%%%%%%%%%%%%%%%%%%%%%%%%%%%%%%%%%%%%%%
%%%%%%%%%%%%%%%%%%%%%%%%%%%%%%%%%%%%%%%%%%%%%%%%%%%
%%%%%%%%%%%%%%%%%%%%%%%%%%%%%%%%%%%%%%%%%%%%%%%%%%%

    %!TEX root = ./ms.tex

%%%%%%%%%%%%%%%%%%%%%%%%%%%%%%%%%%%%%%%%%%%%%%%%%%%
%%%%%%%%%%%%%%%%%%%%%%%%%%%%%%%%%%%%%%%%%%%%%%%%%%%
%%%%%%%%%%%%%%%%%%%%%%%%%%%%%%%%%%%%%%%%%%%%%%%%%%%
\section{Results}
\label{sec:smirnov_results}
\subsection{Smirnov}
\label{smirnov_results}

Switching to a computation which always uses exponentiation dramatically changed the results for very small probabilities, which were too small by a factor of \num{1e55} in some cases.
Whilst these particular probabilities were {\em extremely} \/ small, on the order of \num{e-230}, the underflow phenomenon was already having an impact for other results.

The use of a more accurate function to calculate the powers resulted in significant changes to returned results.
It also required more computation.

Performing the computations in extended or quadruple precision would help with the errors arising from the exponentiation, and roundoff errors in some of the summations.

Performing the computations with an extended range type
(such as a hardware-supported 80 bit extended precision type or 128 bit long doubles type)
would mitigate the underflow/overflow problems to some extent, as the range of exponents is higher for these types (${\pm16381}$ vs ${\pm1021}$).
But the range is not so much greater that the problem is vanquished.

A greater precision would enable greater usage of the Smirnov/Dwass formulation \formularef{Sn_upper}.

In the next 3 tables we shows the statistics on the errors of several algorithms when compared with a SageMath\cite{sagemath} implementation using 300bit real numbers.
The methods reported are the baseline SciPy implementation, then 3 implementations based on \algref{smirnov_alg}.
``double" is an implementation purely using 64-bit doubles;
``double-double" uses double-doubles;
``extended80" uses an 80-bit extended double type;
``double+" using double-doubles for basic addition, multiplication \& division, but the simplified \code{powDSimple}
for the exponentiation.

For the PDF task, the Baseline does not provide a stand-alone function.  Instead it invokes \code{scipy.stats.ksone.pdf()} which
takes a numerical derivative of the \code{scipy.special.smirnov} function.  This approach didn't perform very well, so
a separate implementation of the $\PDF$ based upon \code{scipy.special.smirnov} was created and is labelled ``Baseline (sim)" in the tables.

The testing was performed for about 70 values of $n$ spread across 5 ranges (
$1(1)20$; $20(5)100$; $100(50)1100$; $1100(100)2000$; and $2000(1000) 10000$)
 and 1001 equally spaced $x$-values ($x=0\, (0.001) \,1.0$.)
The statistics were calculated on all pairs where the value calculated by the SageMath implementation exceeded \num{e-275}.

Any relative errors are in units of $\epsilon=2^{-52}$.  I.e. $\frac{1}{\epsilon} \cdot  \frac{\text{correct} - \text{computed}}{\text{correct}}$, or its absolute value.  The value computed using SageMath was rounded-to-nearest and used as a proxy for the correct value.

%%%%%%%%%%%%%%%%%%%%%%%%%%%%%%%%%%%%%%%%%%%%%%%%%%%

\begin{table}[!htb]
\centering
\begin{tabular}{l|rrrr}
{} &          Mean &      Std Dev &         Max\\
Algorithm               &               &              &             \\
\midrule
%Survival Function \\
SF \\
\quad Baseline      &  -4.59e+13 &  2.07e+14 &  4.50e+15 \\
\quad \algabbrevref{smirnov_alg}: double        &   2.71e+02 &  1.74e+03 &  4.28e+04 \\
\quad \algabbrevref{smirnov_alg}: double+       &  -2.06e-02 &  2.84e-01 &  2.44e+00 \\
\quad \algabbrevref{smirnov_alg}: double-double &  -1.00e-04 &  4.78e-02 &  9.99e-01 \\
\quad \algabbrevref{smirnov_alg}: extended80      &   5.00e-04 &  2.05e-01 &  1.65e+00 \\
\midrule
%Probability Density Function \\
PDF \\
\quad Baseline      &  6.56e+250 &       inf &  1.79e+254 \\
\quad \algabbrevref{smirnov_alg}: double        &   2.60e+02 &  1.78e+03 &   1.62e+05 \\
\quad \algabbrevref{smirnov_alg}: double+       &  -2.19e-02 &  3.62e-01 &   2.32e+01 \\
\quad \algabbrevref{smirnov_alg}: double-double &  -1.50e-03 &  5.12e-02 &   1.03e+00 \\
\quad \algabbrevref{smirnov_alg}: extended80      &   1.90e-03 &  2.38e-01 &   2.61e+01 \\
\bottomrule
\end{tabular}
\caption{Smirnov $\SF$ and \PDF: Statistics on the Relative Errors for several algorithms.}
\label{tab:sm_res_relerr_sp}
\end{table}

\tblref{sm_res_relerr_sp} 
shows statistics on the relative errors, in units of $\epsilon=2^{-52}$.
The large values for the Baseline system are mainly due to all the the intermediate underflows causing 0s to be returned.  
As $n \uparrow 1000$, the percentage of computed values with relative error $\approx 1$ increases to about 10\%.
(Recall  that the Baseline computation switched to $\log$ for $n>1013$, avoiding underflow.)
For that system it is better to look at disagreement rates at various tolerances (essentially the survival function for the relative error rate) in \tblref{sm_res_sp}.  After removing the total underflows, there is still plenty of disagreement. 
This is due to underflow effects, and the use of $\log$.
For $1000<n<10000$, the Baseline relative error for the \SF\/ seemed $\approx 3n\epsilon$.)
The relative error for the \SF/\PDF\/ using double-doubles is always less than $1*\epsilon$.  
For the extended80 type, the relative error was almost always less than  $1*\epsilon$, the one observed exception arising from using the Smirnov-Dwass alternate formula with 3 terms.
The relative errors for the "double+" algorithm have a higher mean and variance, but the max is  $<3*\epsilon$, making it a plausible option.
The relative errors for the "double" algorithm have a mean of $\approx 300*\epsilon$, mainly because of the performance of the system for $n>1000$.

The relative errors for the \PDF\/ were larger for all the algorithms, reflecting the higher condition number of the sum.
The ridiculously large errors for the Baseline system arise as there appears to be a lower bound imposed in the returned \PDF, 
so that it always has magnitude larger than \num[group-digits = false]{4.163336e-12}.

For the $\PDF$, the extended80 system was mostly on par with the double-double system, except that it was observed to have a much higher maximum error, almost 20x as high. This difference is even more apparent in \tblref{sm_res_relerr_sp_sqrtn} and will be 
explained below.

\begin{table}[!htb]
\centering
\begin{tabular}{l|rrrrrrr}
\toprule
{} & \multicolumn{7}{c}{Tolerance} \\
{} &      1e-9 & 1e-10 & 1e-11 & 1e-12 &  1e-13 &  1e-14 &  1e-15 \\
Algorithm     &           &       &       &       &        &        &        \\
\midrule
SF \\
\quad Baseline &      1.8\% &  1.9\% &  2.0\% &  5.0\% &  12.2\% &  14.3\% &  30.8\% \\
\midrule
PDF \\
\quad Baseline   &     46.7\% &  50.3\% &  56.8\% &  70.5\% &  73.9\% &  74.2\% &  74.2\% \\
\quad Baseline (sim) &      0.0\% &   0.0\% &   0.0\% &   2.5\% &   6.9\% &  17.9\% &  36.4\% \\
\bottomrule
\end{tabular}
\caption{Smirnov $\SF$ and \PDF: Disagreement rates at several tolerances.}
\label{tab:sm_res_sp}
\end{table}

%%%%%%%%%%%%%%%%%%%%%%%%%%%%%%%%%%%%%%%%%%
%%%%%%%%%%%%%%%%%%%%%%%%%%%%%%%%%%%%%%%%%%
%%%%%%%%%%%%%%%%%%%%%%%%%%%%%%%%%%%%%%%%%%
%%%%%%%%%%%%%%%%%%%%%%%%%%%%%%%%%%%%%%%%%%

\tblref{sm_res_relerr_sp_sqrtn}
shows relative errors
for the $\SF$ and \PDF\/ for $x\in[0, \frac{3}{\sqrt{n}}|$, a reduced domain which covers most $x$ of interest.
In this region, most terms of the summations need to be evaluated and accumulated.  The performance on the \SF\/ is similar to that seen in \tblref{sm_res_relerr_sp}, but the relative errors on the \PDF\/ task have noticeably increased. 

The extended80 system now has maximum error  60x higher than the double-double system.
The cause was determined to be the smaller precision of the significand in the extended80 format: 64 bits vs 106.
For large $n$ and $x \ll \frac{1}{\sqrt{n}}$, the multiplier of $A_j(n, x)$ to construct $D_j(n, x)$ is much larger for the first few values of $j$ than later $j$, and the values of $A_j(n, x)$ themselves are also bigger, with the result that any small  {\it relative}\/ error in $A_j(n, x)$ contributes a large {\it absolute}\/ error to $D_j(n, x)$.
$A_j(n, x)$ needs computation in greater precision than the 64-bits for these small $x$.

%%%%%%%%%%%%%%%%%%%%%%%%%%%%%%%%%%%%%%%%%%%%%%%%%%%

\begin{table}[!htb]
\centering
\begin{tabular}{l|rrr}
\toprule
{} &     Mean &  Std Dev &      Max \\
Algorithm     &          &          &          \\
\midrule
SF \\
\quad Baseline      &   1.312e+03 &  5.055e+03 &  3.584e+04 \\
\quad \algabbrevref{smirnov_alg}: double        &   1.197e+03 &  5.349e+03 &  3.564e+04 \\
\quad \algabbrevref{smirnov_alg}: double+       &  -2.860e-02 &  2.330e-01 &  2.474e+00 \\
\quad \algabbrevref{smirnov_alg}: double-double &  -1.400e-03 &  1.095e-01 &  9.995e-01 \\
\quad \algabbrevref{smirnov_alg}: extended80      &  -1.200e-03 &  1.545e-01 &  1.003e+00 \\
\midrule
PDF \\
\quad Baseline      &   6.477e+10 &  1.035e+12 &  8.654e+14 \\
\quad \algabbrevref{smirnov_alg}: double        &   7.605e+02 &  6.257e+03 &  3.257e+05 \\
\quad \algabbrevref{smirnov_alg}: double+       &  -4.233e-02 &  7.443e-01 &  2.424e+02 \\
\quad \algabbrevref{smirnov_alg}: double-double &  -1.701e-02 &  1.560e-01 &  3.869e+00 \\
\quad \algabbrevref{smirnov_alg}: extended80      &  -3.677e-03 &  4.583e-01 &  2.493e+02 \\
\bottomrule
\end{tabular}
\caption{Smirnov \PDF\/ and \SF: Statistics on the Relative Errors for several algorithms, as measured for $x\in[0, \frac{3}{\sqrt{n}}]$.}
\label{tab:sm_res_relerr_sp_sqrtn}
\end{table}

%%%%%%%%%%%%%%%%%%%%%%%%%%%%%%%%%%%%%%%%%%%%%%%%%%%
%%%%%%%%%%%%%%%%%%%%%%%%%%%%%%%%%%%%%%%%%%%%%%%%%%%
%%%%%%%%%%%%%%%%%%%%%%%%%%%%%%%%%%%%%%%%%%%%%%%%%%%

%%%%%%%%%%%%%%%%%%%%%%%%%%%%%%%%%%%%%%%%%%%%%%%%%%%
%%%%%%%%%%%%%%%%%%%%%%%%%%%%%%%%%%%%%%%%%%%%%%%%%%%
%%%%%%%%%%%%%%%%%%%%%%%%%%%%%%%%%%%%%%%%%%%%%%%%%%%

\subsection{Smirnovi}
\label{smirnovi_results}

The testing was performed for about 40 values of $n$ spread across 4 ranges (
$1(1)10$; $10(10)100$; $100(100)1200$;  and $2000(2000) 10000$)
 and 101 equally spaced $p$-values ($x=0\, (0.01) \,1.0$.)

The methods reported are the baseline SciPy implementation, then 3 implementations based on \algref{smirnovi_alg},
the only difference being which \code{smirnov} function was called to compute the $\SF/\CDF/\PDF$.
``double" is an implementation purely using doubles;
``double-double" uses double-doubles;
``double+" using double-doubles for basic addition, multiplication \& division, but the simplified \code{powDSimple}
for the exponentiation.

The proposed algorithm typically needs 3-4
N-R iterations to achieve convergence within a tolerance of $~\num{2.2e-16}$.
This compares to at least twice that many in SciPy v0.19, which also uses a much higher tolerance of \num{1e-10}.
The maximum number of iterations is also much reduced to about 8, with no failures to converge observed.

\begin{table}[!htb]
\centering
\begin{tabular}{l|rrrrrrr}
\toprule
{} & \multicolumn{7}{c}{Tolerance} \\
{} &      1e-9 & 1e-10 & 1e-11 &  1e-12 &  1e-13 &  1e-14 &  1e-15 \\
Algorithm              &           &       &       &        &        &        &        \\
\midrule
Baseline          &      0.7\% &  0.8\% &  6.6\% &  38.9\% &  79.4\% &  92.9\% &  96.8\% \\
\algabbrevref{smirnovi_alg}: double                 &      0.0\% &  0.0\% &  0.0\% &   0.0\% &   0.0\% &   0.0\% &   0.8\% \\
\algabbrevref{smirnovi_alg}: double+                &      0.0\% &  0.0\% &  0.0\% &   0.0\% &   0.0\% &   0.0\% &   0.4\% \\
\algabbrevref{smirnovi_alg}: double-double                &      0.0\% &  0.0\% &  0.0\% &   0.0\% &   0.0\% &   0.0\% &   0.1\% \\
\bottomrule
\end{tabular}
\caption{Smirnov ISF: Disagreement Rates }
\label{tab:SmirnoviTable}
\end{table}

\tblref{SmirnoviTable} shows the disagreement rates at various tolerances for the Baseline implementation.
It is perhaps not surprising that the Baseline disagreement rates are so much higher than the other systems, given that it used a tolerance of \num{1e-10} as a stopping criterion. One would have expected better from an N-R based root finder, but the inaccuracies in the Baseline \code{smirnov(n,x)} implementation
have an impact on the inversion.

\begin{table}[!htb]
\centering
\begin{tabular}{r|rrr|rrrrr}
\toprule
{} &  Mean &   Std Dev &  Max &  Fail & 1e-10 &  1e-11 &  1e-12 &  1e-14 \\
n              &       &       &      &       &       &        &        &        \\
\midrule
2, $\ldots$ 10      &  12.1 &  12.6 &  273 &  2.1\% &  3.4\% &  24.1\% &  75.6\% &  95.4\% \\
20, $\ldots$ 100    &   7.5 &   3.0 &   48 &  0.1\% &  0.1\% &   2.8\% &  41.7\% &  93.4\% \\
200, $\ldots$ 10000 &   5.2 &   1.2 &   18 &  0.0\% &  0.3\% &   3.5\% &  28.9\% &  91.4\% \\
\bottomrule
\end{tabular}
\caption{Smirnov ISF: Iterations and Disagreement- Baseline system.}
\label{tab:SmirnoviTableSciPy}
\end{table}

\begin{table}[!htb]
\centering
\begin{tabular}{r|rrr|rr}
\toprule
{} &          Mean &  Std Dev & Max &  Fail & 1e-14 \\
n              &               &      &     &       &       \\
\midrule
2, $\ldots$ 10      &           4.1 &  1.0 &   6 &  0.0\% &  0.0\% \\
20, $\ldots$ 100    &           3.9 &  0.6 &   5 &  0.0\% &  0.0\% \\
200, $\ldots$ 10000 &           3.1 &  0.5 &   4 &  0.0\% &  0.0\% \\
\bottomrule
\end{tabular}
\caption{Smirnov ISF: Iterations and Disagreement - \algref{smirnovi_alg}}
\label{tab:SmirnoviTableAlgX}
\end{table}

 \tblref{SmirnoviTableSciPy} and \tblref{SmirnoviTableAlgX} shows statistics for the computations, broken out by small, mid-size and large $n$.
 The mean and standard deviations are calculated over the values of $n, p$ which do not exceed 500 iterations.
 The failure is the average percentage of values of $p$ that exceed 500 iterations.
 The tolerance columns list the average percentage of values returned whose relative error exceeded the specified tolerance.

It becomes clear that the small values of $n$ are more problematic for the Baseline than larger $n$.
The disagreement rates are much higher for $n\leq 10$, as are the convergence failure rates.
Interestingly the number of iteration required decreases as $n$ increases, presumably due to the more accurate nature of the initial estimate $X_0$, both for the Baseline and \algref{smirnovi_alg}.
This is helpful as each iteration involves a computation of $S_n(x)$, which takes $O(n)$ time.

The performance of \algref{smirnovi_alg} on very small probabilities ($p_{SF}=2^{-n}, n=100 (100) 1023$) was also evaluated.
It was observed that the number of iterations was often higher and much more variable.
The cause is two-fold.
For such $p$, the corresponding value of $x$ is much greater than $\frac{1}{\sqrt n}$, perhaps even close to 1.
In that part of the domain,
the asymptote lies above (and is much greater than) the actual $\SF$, so
 the initial estimate $X_0$ is greater than the root (hence on the wrong side) but not that close to it.
 That can be addressed somewhat by replacing $X_0$ with a value closer to $A$, the left end of the bracket.
 But the real issue is that the ratio
 $\frac{S_n''(x)}{2S_n'(x)}$ is very large, so that the \ordinalnum{2} and higher order terms in the Taylor Series dominate the \ordinalnum{1} derivative term.
 If the current estimate $X_n$ is greater than the root, then the next one $X_{n+1}$ is likely to overshoot the current bracket, and bisection steps are performed until the iterate is less than the root.
 Even if $X_n < x$, the deltas in the N-R step $\frac{f(X_n)-p_{SF}}{f'(X_n)}$ do not decrease fast enough, so bisections take over.
 We observe a chain of mostly bisection steps until the desired tolerance is achieved.
 For $n=400, p_{SF}=2^{-500}, x\approx0.624$, $\frac{S_n''(x)}{2S_n'(x)}\approx -627$ and it takes about 20 iterations to sufficiently converge.
 An even more extreme example is $n=500, p_{SF}=2^{-1023}, x\approx0.765738635666$, with $\frac{S_n''(x)}{2S_n'(x)}\approx  \num{-2.5e+08}$ which takes about 50 iterations to converge, mostly bisection steps.
Replacing bisection steps with secant steps (or just leaving as N-R steps) doesn't help either --- the curvature of $f$ is too large.

While the stopping criterion of the algorithm is a relative tolerance of $~\num{2.2e-16}$ in successive iterates, this is a little misleading.
The computation of $\smirnov$ itself is only that accurate for some $(n, x)$ values.
Much depends on the accuracy and precision of the $\code{pow}$ function.
That can also determine for which $x$ to use \formularef{Sn_upper}.
For $x\leq\frac{1}{n}$, \formularef{Sn_upper1} is definitely usable, for computing the $\CDF$/$\SF$ and inverse $\CDF$/$\SF$ -- it is self-consistent and protects the N-R iteration steps from the tricky endpoint.
Beyond that, one is left with determining if there is a cutoff that provides a seamless transition from one to the other.
Here we chose to stick to $nx\leq1$ for pure double implementations,  and $nx\leq3$ for double-double implementations.

Other common root-finding algorithms could be used.
Given a tight starting bracket, Brent's method averaged 6 iterations, Sidi's method \cite{Sidi2008a} (with $k=2$) and False Position with Illinois  both averaged about 5 iterations.
If the starting bracket is wider, the number of iterations required for very small values of $p_{SF}$ becomes much larger, due to the extreme flatness of the curve at the right hand end point.

An argument could be made that N-R requires two function evaluations per iteration, for $S_n(x)$ and $S_n'(x)$, so that just counting iterations is underestimating the N-R work.  While this is true, the computation of the two quantities in this work is combined into one function which calculates both with little incremental effort over the effort of computing just one, reusing computations from one for the other.

%%%%%%%%%%%%%%%%%%%%%%%%%%%%%%%%%%%%%%%%%%%%%%%%%%%
%%%%%%%%%%%%%%%%%%%%%%%%%%%%%%%%%%%%%%%%%%%%%%%%%%%
%%%%%%%%%%%%%%%%%%%%%%%%%%%%%%%%%%%%%%%%%%%%%%%%%%%

%%%%%%%%%%%%%%%%%%%%%%%%%%%%%%%%%%%%%%%%%%%%%%%%%%%
%%%%%%%%%%%%%%%%%%%%%%%%%%%%%%%%%%%%%%%%%%%%%%%%%%%
%%%%%%%%%%%%%%%%%%%%%%%%%%%%%%%%%%%%%%%%%%%%%%%%%%%

    \section{Summary}

The \CDF/\SF \/ is a sum of many terms, each term is a product of expressions that are prone to underflow/overflow or are hard to compute accurately. 
We showed how to rework the computation to avoid underflow/overflow, by being more careful about computing $z^m$, and selectively using extra precision for some computations.

Using Newton-Raphson successfully to compute the \ISF \/ requires a good initial estimate, and a good approximation to the derivative.  
The contributions of the higher order terms in the Taylor series expansions for the quantiles can make the root-finding a little problematic.  
We showed how to generate a narrow interval enclosing the root, a good starting value for the iterations, and a 
way to calculate the derivative with little additional work, so that many fewer N-R iterations are required and the 
computed values have smaller errors.

\iftoggle{usingbibtex}{
\bibliographystyle{unsrtnat}
\bibliography{library}}
{
\sloppy
\printbibliography
}

\end{document}